\documentclass[sigconf]{aamas}  

\usepackage{booktabs}
\usepackage{subfiles}
\usepackage{algorithmic}

\newcommand{\ve}[1]{\boldsymbol{#1}}

\newcommand{\name}{\text{SIC}}
\newcommand{\signal}{\text{coordination signal}}

\usepackage[normalem]{ulem}

\usepackage{amsmath,amsthm,amssymb,bm}
\usepackage{subcaption,diagbox}
\usepackage{booktabs}
\usepackage{algorithm}
\usepackage{tikz}
\usepackage{flushend}

\begin{document}

\title{Signal Instructed Coordination in Cooperative Multi-agent Reinforcement Learning}  


\author{Liheng Chen\textsuperscript{1,2$\dagger$}, Hongyi Guo\textsuperscript{1}, Yali Du\textsuperscript{3}, Fei Fang\textsuperscript{4$\ddagger$}, Haifeng Zhang\textsuperscript{3$\mathsection$}, Yaoming Zhu\textsuperscript{1}, \\
	Ming Zhou\textsuperscript{1}, Weinan Zhang\textsuperscript{1}, Qing Wang\textsuperscript{5}, Yong Yu\textsuperscript{1}}
\affiliation{%
	\institution{\textsuperscript{\rm 1}Shanghai Jiao Tong University, \textsuperscript{\rm 2}Tencent AI Lab, \textsuperscript{\rm 3}University College London, \\
		\textsuperscript{\rm 4}Carnegie Mellon University, \textsuperscript{\rm 5}Huya AI\\
		\textsuperscript{$\dagger$}clhbob@sjtu.edu.cn, 
		\textsuperscript{$\ddagger$}feif@cs.cmu.edu, 
		\textsuperscript{$\mathsection$}haifeng.zhang@ucl.ac.uk
	}
}

\renewcommand{\shortauthors}{Liheng Chen et al.}

%
%


%
%
%
%
%

\begin{abstract}  
In many real-world problems, a team of agents need to collaborate to maximize the common reward.
Although existing works formulate this problem into a centralized learning with decentralized execution framework, 
which avoids the non-stationary problem in training, 
their decentralized execution paradigm limits the agents' capability to coordinate.
Inspired by the concept of correlated equilibrium,
we propose to introduce a \textit{$\signal$} to address this limitation, 
and theoretically show that following mild conditions,
decentralized agents with the $\signal$ can coordinate their individual policies as manipulated by a centralized controller.
The idea of introducing $\signal$ is to encapsulate coordinated strategies into the signals,
and use the signals to instruct the collaboration in decentralized execution.
To encourage agents to learn to exploit the $\signal$,
we propose \textit{Signal Instructed Coordination} ($\name$), a novel coordination module that can be integrated with most existing MARL frameworks.
SIC casts a common signal sampled from a pre-defined distribution to all agents, 
and introduces an information-theoretic regularization to facilitate the consistency between the observed signal and agents' policies. 
Our experiments show that SIC consistently improves performance over well-recognized MARL models in both matrix games and a predator-prey game with high-dimensional strategy space.
\end{abstract}

\maketitle


\section{Introduction} \label{sec: intro}
Multi-agent interactions are common in real-world scenarios such as traffic control \cite{Nunes:2004:LMS:1018411.1018879}, smartgrid management \cite{Schneider:1999:DVF:645528.657645}, 
packet routing in networks \cite{weihmayer1994application}, 
and social dilemmas \cite{leibo2017multi}.
Motivated by these applications and inspired by the success of deep reinforcement learning (DRL) in single-agent settings \cite{mnih2015human}, 
there is a growing interest in deep multi-agent reinforcement learning (MARL)~\cite{lowe2017multi,foerster2018counterfactual,foerster2017stabilising,jaderberg2018human}, which studies how the agents can learn to act strategically by adopting RL algorithms.

In cooperative multi-agent environments, 
a straightforward approach is the \textit{fully centralized} paradigm, 
where a centralized controller is used to make decisions for all agents, 
and its policy is learned by applying successful single-agent RL algorithms.
However, the fully centralized method assumes an unlimited communication bandwidth, which is impractical in many real-world scenarios.
Besides, it suffers from exponential growth of the size of the joint action space with the number of agents.
Therefore, decentralized execution approaches are proposed, including the \textit{fully decentralized} paradigm and the \textit{centralized learning with decentralized execution} (CLDE) \cite{oliehoek2008optimal,lowe2017multi} paradigm.
The fully decentralized method models each participant as an individual agent with its own policy and critic conditioned on local information.
This setting fails to solve the non-stationary environment problem \cite{lanctot2017unified,matignon2012independent}, 
and is empirically deprecated by \cite{foerster2016learning,li2018deep}.
In CLDE framework, 
agents can leverage global information including the joint observations and actions of all agents in the training stage, e.g., through training a centralized critic, 
but the policy of an agent can only be dependent on the individual information and thus they can behave in the decentralized way in the execution stage.
This training paradigm bypasses the non-stationary problem~\cite{foerster2016learning,li2018deep}, 
and can lead to some coordination among the cooperative agents empirically~\cite{lowe2017multi}.

Despite the merits of CLDE, 
the feasible joint policy space with distributed execution is much smaller than the joint policy space with a centralized controller, 
limiting the agents' capability to coordinate.
For example, 
in a two-agent traffic system with agents A and B, 
whose individual action space is \{go, stop\}, 
we cannot find a joint policy that satisfies $P(\text{A goes \& B stops})=P(\text{A stops \& B goes})=0.5$
if both agents are making decisions independently.
Previous works \cite{sukhbaatar2016learning,peng2017multiagent,jiang2018learning,iqbal2018actor} adopts peer-to-peer communication mechanism to facilitate coordination, 
but they require specially designed communication channels to exchange information and the agents' capability to coordinate is limited by the accessibility and the bandwidth of the communication channel.

Inspired by the \textit{correlated equilibrium} (CE) \cite{leyton2008essentials,aumann1974subjectivity} concept in game theory, 
we introduce a \textit{$\signal$} to allow for more correlation of individual policies and to further facilitate coordination among cooperative agents in decentralized execution paradigms.
The $\signal$ is conceptually similar to the signal sent by a correlation device to induce CE.
It is sampled from a distribution at the beginning of each episode of the game and carries no state-dependent information. 
After observing the same signal, 
different agents learn to take corresponding individual actions to formulate an optimal joint action.
Such $\signal$ is of practical importance. 
For example, the previous traffic system example can introduce a traffic policeman that sends a public signal via his pose to each agent.
The type of the pose may be dependent on the current time (state-free) as a traffic light is, but agents can still coordinate their actions without any explicit communication among them. In addition, we prove that for a group of fully cooperative agents, if the signal's distribution satisfies some mild conditions, 
the joint policy space is equal to the centralized joint policy space.
Therefore, the coordination signal expands the joint policy space while still maintains the decentralized execution setting,
and is helpful to find a better joint policy.

To incentivize agents to make full use of the coordination signal, 
we propose \textit{Signal Instructed Coordination} ($\name$), 
a novel plug-in module for learning coordinated policies. 
In SIC, 
a continuous vector is sampled from a pre-defined normal distribution as the coordination signal, 
and every agent observes the vector as an extra input to its policy network. 
We introduce an information-theoretic regularization, which maximizes the mutual information between the signal and the resulting joint policy. 
We implement a centralized neural network to optimize the variational lower bound \cite{barber2003algorithm,chen2016infogan,li2017infogail} of the mutual information. 
The effects of optimizing this regularization are three-fold: 
it (i) encourages each agent to align its individual policy with the $\signal$,
(ii) decreases the uncertainty of policies of other agents to alleviate the difficulty to coordinate, 
and (iii) leads to a more diverse joint policy.
Besides, SIC can be easily incorporated with most models that follow the decentralized execution paradigm, 
such as MADDPG \cite{lowe2017multi} and COMA \cite{foerster2018counterfactual}.

To evaluate SIC, we first conduct insightful experiments on a multiplayer variant of matrix game \textit{Rock-Paper-Scissors-Well} \cite{rpsw} to demonstrate how $\name$ incentivize agents to coordinate in both one-step and multi-step scenarios. 
Then we conduct experiments on \textit{Predator-Prey}, 
a classic game implemented in multi-agent particle worlds \cite{lowe2017multi}. 
We empirically show that by adopting SIC, agents learn to coordinate by interpreting the signal differently and thus achieve better performance. 
Besides, the visualization of the distribution of collision positions in \textit{Predator-Prey} provides evidence that SIC improves the diversity of policies. 
An additional parameter sensitivity analysis manifests that SIC introduces stable improvement.

\section{Methods}
In this section, we formally propose our model. 
We first provide the background for multi-agent reinforcement learning. 
Then we analyze the superiority of signal instructed approach over previous paradigms. 
Finally, we introduce \textit{Signal Instructed Coordination} ($\name$) 
and its implementation details.

\subsection{Preliminaries}
We consider a fully cooperative multi-agent game with $N$ agents.
The game can be described as a tuple as $\langle \mathcal{I},\mathcal{S}, \mathcal{A},T,r,\gamma,\rho_0 \rangle$. 
Let $\mathcal{I}=\{1,2,\cdots,n\}$ denote the set of $n$ agents. 
$\mathcal{A}=\langle \mathcal{A}_1, \mathcal{A}_2, \cdots, \mathcal{A}_n\rangle$ is the joint action space of agents, 
and $\mathcal{S}$ is the global state space.
At time step $t$, 
the group of agents takes the joint action $\mathbf{a}_t=\langle a_{1t}, a_{2t}, \cdots, a_{nt}\rangle$ with each $a_{it}\in \mathcal{A}_i$ indicating the action taken by the agent $i$.
$T(s_{t+1}|s_t,\mathbf{a}_t):\mathcal{S}\times \mathcal{A}\times \mathcal{S}\rightarrow[0,1]$ is the state transition function. 
$r(s_t,\mathbf{a}_t):\mathcal{S}\times \mathcal{A}\rightarrow\mathbb{R}$ indicates the reward function from the environment. 
$\gamma\in[0,1)$ is a discount factor and $\rho_0:\mathcal{S}\rightarrow[0, 1]$ is the distribution of the initial state $s_0$. 

Let $\pi_i(a_{it}|s_t):\mathcal{S}\times \mathcal{A}_i\rightarrow[0,1]$ be a stochastic policy for agent $i$, 
and denote the joint policy of agents as $\bm{\pi}=\langle \pi_1,\pi_2,\cdots,\pi_n\rangle \in\Pi$ where $\Pi$ is the joint policy space.
Let $J(\bm{\pi})=\mathbb{E}_{\bm{\pi}}\left[\sum_{t=0}^{\infty}\gamma^t r_{t}\right]$ denotes the expected discounted sum of rewards, 
where $r_t$ is the reward received in time-step $t$ following policy $\bm{\pi}$. 
We aim to optimize the joint policy $\bm{\pi}$ to maximize $J(\bm{\pi})$. 

In this paper, we also consider environments with opponent agents. 
However, we focus on promoting cooperation among the controllable agents and regard the opponent agents as a part of the environment. 

\subsection{Joint Policy Space with Coordination Signal}
In the fully centralized paradigm, a centralized controller is used to manipulate a group of agents.
We denote $\Pi^C$ as the policy space of the centralized controller and $\bm{\pi}^C:\mathcal{S} \times \mathcal{A} \rightarrow [0,1]$ as a joint policy in $\Pi^C$. 
In the decentralized execution paradigm,
the agents make decisions independently according to their individual policies $\pi_i^D:\mathcal{S} \times \mathcal{A}_i \rightarrow [0,1]$.
We define the policy space of agent $i$ as $\Pi_i^D$ and the joint policy space as $\Pi^D=\Pi_1^D\times\Pi_2^D\times\cdots\times\Pi_n^D$, i.e. the Cartesian product of the policy spaces of each agent. For a joint policy $\bm{\pi}^D \in \Pi^D$, we have  
$\bm{\pi}^D(\bm{a}|s)=\pi_1^D(a_1|s)\cdot\pi_2^D(a_2|s)\cdot\cdots\cdot\pi_n^D(a_n|s)$, $\forall s \in \mathcal{S}$ and $\forall \bm{a}=\langle a_1, a_2, \cdots, a_n\rangle \in\mathcal{A}$.
We conclude the relation between $\Pi^C$ and $\Pi^D$ as the following theorem.

\begin{theorem}
    $\Pi^D$ is a subset of $\Pi^C$.
\end{theorem}
\begin{proof}
    For $\forall\bm{\pi}^D\in\Pi^D$, 
    we can construct $\bm{\pi}^C\in\Pi^C$, 
    which suffices that for $\forall s\in \mathcal{S}$, $\forall \bm{a}=\langle a_1, a_2, \cdots, a_n\rangle \in \mathcal{A}$,  
    \begin{equation}
        \bm{\pi}^C(\bm{a}|s)=\pi^D_1(a_1|s)\cdot\pi^D_2(a_2|s)\cdot\cdots\cdot\pi^D_n(a_n|s)=\bm{\pi}^D(\bm{a}|s),\nonumber
    \end{equation}
    i.e. $\bm{\pi}^C=\bm{\pi}^D$.
    Therefore, we have $\bm{\pi}^D\in\Pi^C$ and $\Pi^D \subseteq \Pi^C$.
\end{proof}
    We use a counterexample to show that 
    not every joint policy in $\Pi^C$ is an element of $\Pi^D$.
    In a two-agent system where each agent has two actions $x$ and $y$, 
    $\forall s\in\mathcal{S}$, 
    $\exists \bm{\pi}^C\in\Pi^C$ that suffices
    \begin{equation}
        \bm{\pi}^C(\langle a_1=x, a_2=x\rangle|s)=\bm{\pi}^C(\langle a_1=y, a_2=y\rangle|s)=0.5, \nonumber
    \end{equation}
    but $\bm{\pi}^C \notin \Pi^D$,
    because there is no valid solution for 
    \begin{eqnarray}
        &&\pi_1(a_1=x|s) \cdot \pi_2(a_2=x|s) \nonumber\\
        &=&\pi_1(a_1=y|s) \cdot \pi_2(a_2=y|s) \nonumber\\
        &=&(1-\pi_1(a_1=x|s)) \cdot (1-\pi_2(a_2=x|s))\nonumber \\
        &=&0.5. 
    \end{eqnarray}
Since $\Pi^C$ includes $\Pi^D$,
the best joint policy in $\Pi^C$ is superior or equal to the best joint policy in $\Pi^D$.
However, due to computational complexity concerns, 
the decentralized execution paradigm is more practical in large-scale environments.
Taking into account the superiorities of both centralized and decentralized paradigms, we have motivation to propose a new framework which has centralized policy space $\Pi^C$ and is executed in a decentralized way.

\begin{figure}[t!]
	\centering
	\includegraphics[width=0.48\textwidth]{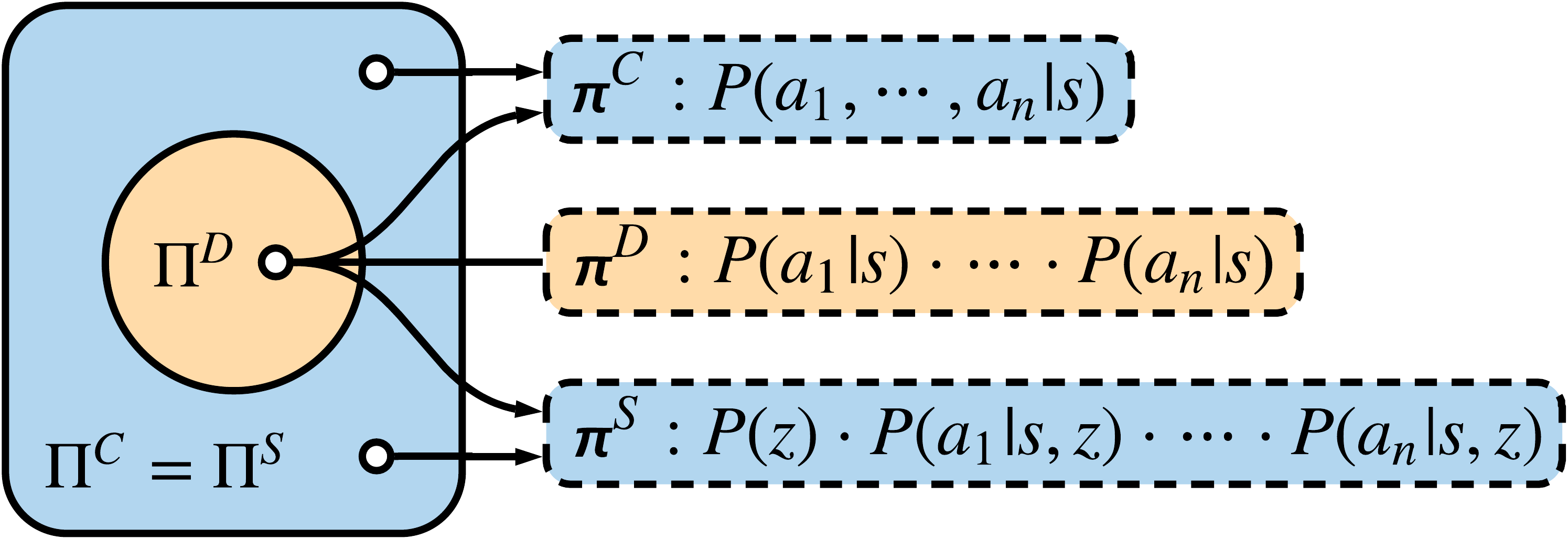}
	\caption{The relationship among $\Pi^C$, $\Pi^D$ and $\Pi^S$ is $\Pi^C=\Pi^S\supseteq\Pi^D$.
	The white circle represents an element in the set.}
	\label{fig:ne-utility-relationship}
\end{figure}

From a game-theoretic perspective, 
the decentralized agents try to reach a Nash equilibrium (NE),
with each individual policy as a best response to others' policies.
Previous studies on computational game theory show that by following the signal provided by a \textit{correlation device}, 
agents may reach a more general type of equilibrium, \textit{correlated equilibrium} (CE) \cite{leyton2008essentials},
which can potentially lead to better outcomes for all agents \cite{aumann1974subjectivity, jiang2011polynomial,farina2019correlation,ashlagi2008value}.
Inspired by CE, 
we propose a \textit{signal instructed} framework,
which provides a larger equilibrium space and maintains the decentralized execution scheme for ease of training.
We introduces a coordination signal sent to every agent at the beginning of a game,
which is conceptually close to the signal sent by the correlation device in CE.

The usage of the signal changes $\Pi^C$ to a different joint policy space, $\Pi^S$.
Every agent observes the same signal $z\in\mathcal{Z}$ sampled from a distribution $P_z$,
where $\mathcal{Z}$ is the \textit{signal space}, 
and learns an individual policy as $\pi_i^S(a_i|s, z)$.
Therefore, 
the agents formulate a special joint policy, 
$\bm{\pi}^S$, 
which suffices that
$\forall s \in \mathcal{S}$, 
$\forall \bm{a}=\{a_1, a_2, \cdots, a_n\}\in\mathcal{A}$,
and $\forall z\in\mathcal{Z}$,
$\bm{\pi}^S(\bm{a}, z|s)=P_z(z)\cdot\pi_1^S(a_1|s, z)\cdot\pi_2^S(a_2|s, z)\cdot\cdots\cdot\pi_n^S(a_n|s, z)$.
$\bm{\pi}^S$ is a conditional joint distribution of $\bm{a}$ and $z$,
and differs from aforementioned types of joint policies.
However, 
by regarding $z$ as an extension of global state, 
we do not change the way we model the individual policy of each agent, 
which is still $\pi_i^S(a_i|s')$ with $s'=(s, z)$.

In signal instructed approach, 
all agents observe the same $z$, 
and we assume that every agent follows the instruction of $z$, i.e.,
takes only one specific corresponding action $a^z_i$.
In other words, 
every agent's policy is
\begin{equation}
    \pi_i(a_i|s, z)=\begin{cases}
    1 & a_i = a_i^z \\
    0 & \text{otherwise}.
    \end{cases} \nonumber
\end{equation}
We denote the corresponding joint action as $\bm{a}^z=\{a_1^z, a_2^z, \cdots, a_n^z\}$.
Intuitively, this assumption is like that agents make an ``agreement'' on which joint action to take in current state when observing $z$, 
which is common in real-world scenarios.
For example, in a traffic junction,
a traffic policeman standing in the center directs the traffic by posing to cars in all directions.
The cars can tell whether they should accelerate or stop from the same observed pose, 
even if the policeman changes his pose in accord to time (state-free).
The agents can be regarded as reaching a CE.
Following the assumption still results in a stochastic joint policy,
with the stochasticity conditioned on $z$ now.
With this assumption, we derive Theorem \ref{them: pisic=pim}.
\begin{theorem} \label{them: pisic=pim}
	$\Pi^S$ is equal to $\Pi^C$.
\end{theorem}
\begin{proof}
    We prove this theorem in two steps:
    \begin{enumerate}
    \item $\Pi^S\supseteq\Pi^C$: 
    $\forall\bm{\pi}^C\in\Pi^C$,
    we can construct $\bm{\pi}\in\Pi^S$, 
    which suffices that $\forall s\in \mathcal{S}$, $\forall \bm{a}=\{a_1, a_2, \cdots, a_n\}\in \mathcal{A}$, we assign a signal $z\in \mathcal{Z}$ to $\bm{a}$ with $P_z(z)=\bm{\pi}^C(\bm{a}|s)$, s.t.
    \begin{eqnarray}
        \bm{\pi}(z, \bm{a}|s) &=& P_z(z)\bm{\pi}(\bm{a}|z, s) \nonumber \\
        & = & P_z(z)[\pi_1(a_1|s, z)\cdot\pi_2(a_2|s, z)\cdot\cdots
        \pi_n(a_n|s, z)] \nonumber\\
        &=& P_z(z)[1\cdot1\cdot\cdots\cdot1] = \bm{\pi}^C(\bm{a}|s)\nonumber
    \end{eqnarray}
    Therefore, we have $\bm{\pi}=\bm{\pi}^C\in\Pi^C$.
    Since every joint policy in $\Pi^C$ is also an element of $\Pi^S$, 
    we conclude that $\Pi^S\supseteq\Pi^C$.
    \item $\Pi^S\subseteq\Pi^C$: 
    $\forall\bm{\pi}^S\in\Pi^S$, 
    we can construct $\bm{\pi}\in\Pi^C$, 
    which suffices that $\forall s\in \mathcal{S}$, $\forall \bm{a}=\{a_1, a_2, \cdots, a_n\}\in \mathcal{A}$, 
    $\forall z\in \mathcal{Z}$,  
    \begin{equation}
        \bm{\pi}(\bm{a}|s) =
        \begin{cases}
        P_z(z) & \bm{a}=\bm{a}^z \\
        0 & \text{otherwise}.
        \end{cases} \nonumber
    \end{equation}
    Therefore, we have $\bm{\pi}=\bm{\pi}^S\in\Pi^S$.
    Since every joint policy in $\Pi^S$ is also an element of $\Pi^C$, 
    we conclude that $\Pi^S\subseteq\Pi^C$.
    \end{enumerate}
    Given $\Pi^S\subseteq\Pi^C$ and $\Pi^S\supseteq\Pi^C$, 
    we conclude that $\Pi^S=\Pi^C$.
\end{proof}
This theorem shows that the signal instructed method enlarged the joint policy space to the same size as $\Pi^C$.
Therefore, agents can still act with their decentralized individual policies,
while exploiting a larger joint policy space.
Note that if agents do not follow the signal and take stochastic individual policies, 
$\Pi^S$ may be even larger due to the additional variable $z$.
However, this often results in miscoordination of agents and is undesired.
We illustrates the relationship among different joint policy spaces in Fig. \ref{fig:ne-utility-relationship}.

A practical concern is that according to the assumption, 
agents need to assign every joint action to a specific $z$,
which means that the size of $\mathcal{Z}$ is as large as $\mathcal{A}$.
Fortunately, 
the set of optimal joint actions is often small, 
and hence an optimal $\bm{\pi}^S$ needs only a small subspace of $\mathcal{Z}$ to instruct agents to take those optimal joint actions.
Besides, we can lower the complexity with neural networks as demonstrate in Sec. \ref{sec: parameter sensitivity}.

In addition to our signal instructed method, 
previous works \cite{sukhbaatar2016learning,peng2017multiagent,jiang2018learning,iqbal2018actor} also propose to introduce the communication mechanism, 
which encourages agents to coordinate in the decentralized execution approach by exchanging communication vectors.
However, in communication agents have to execute three tasks to accomplish successful coordination:
sending meaningful signals (signaling), 
receiving the right signal (receiving),
and interpreting it correctly (listening).
\cite{jiang2018learning} points out that the receiving process suffers from the \textit{noisy channel} problem, 
which arises when all other agents use the same communication channel to simultaneously send information to one agent.
The agent needs to distinguish useful information from useless or irrelevant noise.
In addition, \cite{lowe2019pitfalls} shows that successful signaling does not necessarily lead to effective listening.
Therefore, it is not easy to explore better joint policy through communication in complex scenarios.
Unlike these peer-to-peer communication approaches, 
the signal instructed approach improves coordination in a top-to-bottom way.
It circumvents signaling and receiving issues,
and focuses on listening stage as we will discuss in the next section. 

\subsection{Signal Instructed Coordination} \label{sec: sic}
When a $\signal$ is observed, 
how to incentivize agents to follow its instruction and coordinate is a critical issue. 
As a $\signal$ is sampled from $P_z$ and carries no state-dependent information, 
it is possible for agents to treat it as random noise and ignores it during the training process.
Our idea is to facilitate the $\signal$ to be entangled with agents' behaviours and thus encourage the coordination in execution. 
We name our method as \textit{Signal Instructed Coordination  ($\name$)}.
SIC introduces an information-theoretic regularization to ensure the signal makes an impact in agents' decision making.
This regularization aims to maximize the mutual information between the signal, $z$, and the joint policy, $\bm{\pi}^S$, given current state, $s$, as
\begin{eqnarray}
I(z; \bm{\pi}^S(\bm{a}, z|s)) & = &- H(\bm{\pi}^S|z) +  H(\bm{\pi}^S)  \nonumber\\
& = & - H(\pi^S_i, \bm{\pi}^S_{-i}|z) + H(\bm{\pi}^S)  \nonumber\\
& = & - H(\pi^S_i|z)  - H(\bm{\pi}^S_{-i}|\pi^S_i, z) + H(\bm{\pi}^S), \label{eq: original mutual information}
\end{eqnarray}
where $\bm{\pi}^S$, $\pi^S_i$ and $\bm{\pi}^S_{-i}$ are abbreviations for $\bm{\pi}^S(\bm{a}, z|s)$, $\pi^S_i(a_{i}, z|s)$ and $\bm{\pi}^S_{-i}(\bm{a}_{-i}, z|s)$, and $\bm{\pi}^S_{-i}$ and $\bm{a}_{-i}$ are the joint policy and the joint action of all agents except agent $i$ respectively.
The decomposition of $\bm{\pi}^S$ into $\pi^S_i$ and $\bm{\pi}^S_{-i}$ in the second line holds in our decentralized approach.

Through decomposing the regularization term in Eq. \eqref{eq: original mutual information}, 
one can find that the effects for minimizing the mutual information between signal and policy are threefold. 
Minimizing the first term increases consistency between the coordination signal and the individual policy to suffice the assumption of Theorem \ref{them: pisic=pim}. 
Minimizing the second term ensures low uncertainty of other agents' policies,
which is beneficial to establish coordination among agents. 
Maximizing third term encourages the joint policy to be diverse, which prohibits the opponents from inferring our policy in competition.
These effects in combination improves the performance of the joint policy.
However, directly optimizing Eq. \eqref{eq: original mutual information} is troublesome in implementation.
Considering the symmetry property of mutual information, we aim to maximize
\begin{eqnarray}
I(z; \bm{\pi}^S(\bm{a}, z|s)) & = & - H(z|\bm{\pi}^S(\bm{a}, z|s)) + H(z) \nonumber\\
& = & \mathbb{E}_{z\sim P_z, \bm{a}\sim\bm{\pi}(\bm{a}|s, z)}[\mathbb{E}_{z'\sim P(\cdot|s, \bm{a})}\log P(z'|s, \bm{a})] \nonumber \\
&& + H(z), \label{eq: mutual information}
\end{eqnarray}
where $\bm{\pi}(\bm{a}|s, z)$ is the Cartesian product of individual policies after observing a specific $z$, 
and $P(\cdot|s, \bm{a})$ is the posterior distribution estimating the probability of a specific signal $z'$ after seeing the state $s$ and the joint action $\bm{a}$.
Note that $P(\cdot|s, \bm{a})$ is not the same as $P_z$.
Since we have no knowledge of the posterior distribution,
we circumvent it by introducing a variational lower bound \cite{barber2003algorithm,chen2016infogan,li2017infogail} which defines an auxiliary distribution $U(\cdot|s, \bm{a})$ as:
\begin{eqnarray}
\text{Eq.}~\eqref{eq: mutual information} & = & \mathbb{E}_{z\sim P_z, \bm{a}\sim\bm{\pi}(\bm{a}|s, z)}[D_{KL}(P(\cdot|s, \bm{a})||U(\cdot|s, \bm{a})] + \nonumber \\ 
&&\mathbb{E}_{z\sim P_z, \bm{a}\sim\bm{\pi}(\bm{a}|s, z)}[\mathbb{E}_{z'\sim P(\cdot|s, \bm{a})}\log U(z'|s, \bm{a})]  + H(z) \nonumber \\
&\geq& \mathbb{E}_{z\sim P_z, \bm{a}\sim\bm{\pi}(\bm{a}|s, z)}[\mathbb{E}_{z'\sim P(\cdot|s, \bm{a})}\log U(z'|s, \bm{a})], \label{eq: lower bound}
\end{eqnarray}
where $U(z|s, \bm{a})$ is an approximation of $P(z|s, \bm{a})$. 
The inequality operator in the last line holds due to the non-negative property of KL divergence and entropy as $D_{KL}(\cdot)\geq 0$ and $H(\cdot)\geq 0$.
Therefore, we derive a mutual information loss (MI loss) from Eq. \eqref{eq: lower bound} as
\begin{equation}
L_I(\bm{\pi}, U) = -\mathbb{E}_{z\sim P_Z, s\sim \tau, \bm{a}\sim\bm{\pi}(\bm{a}|s, z)}\log U(z|s, \bm{a}). \label{eq: mi loss}
\end{equation}
where $\tau$ is trajectories of game episodes.
Minimizing Eq. \eqref{eq: mi loss} facilitates agents to follow the instruction of the coordination signal.

Another important issue is how to model $P_z$.
Using a discrete signal space may raise differentiable problems in back-propagation and increase model complexity when encoding signals as one-hot vectors.
Instead, we propose to adopt a continuous signal space,
and approximate $P_z$ in a Monte Carlo way.
In detail, 
we sample a $D_z$-dimension continuous vector $v$ from a normal distribution,
$\mathcal{N}(\ve{0}, I^{D_z\times D_z})$,
and distribute it to all agents.
$v$ can be viewed as a value of the variable $z$.
Agents learn to divide the $R^{D_z}$ space into several subspaces,
with each corresponding to one signal and hence one optimal joint action.
The probability of sampling a specific $z$, $P_z(z)$, 
is approximated by the probability of sampling a vector $v$ that belongs to the corresponding subspace.

To compute $U(z|s, \bm{a})$,
we use a centralized multi-layer feed-forward network, named as U-Net, as a parameterized function, $f_U$.
U-Net inputs $s$ and $\bm{a}$, 
and outputs a continuous vector with the same dimension as $z$ as the reconstructed signal $z'=f_U(s, \bm{a})$.
$U(z|s, \bm{a})$ is measured by the mean squared error between $z'$ and $z$.
One obvious advantage of SIC is that it can be easily integrated with most existing models with policy networks,
as shown in Fig. \ref{fig:sic}.
We expect to optimize parameters of policy networks and U-net concurrently by minimizing $L_I$,
so that the learned joint action is an optimal one.
However, if agents adopt stochastic policies as in COMA, 
it is challenging to pass gradients to policy networks through sampled $\bm{a}$.
Therefore, we use a simplified approximation $U(z|s, \bm{h})$, 
where $\bm{h}$ is the concatenation of last layers of hidden vectors in policy networks.
Parameters of the centralized U-net, $\omega$, and parameters of decentralized policies, $\theta=\langle \theta_1, \theta_2, \cdots, \theta_n\rangle $, 
are jointly optimized as:
\begin{eqnarray}
&&\max_{\omega, \theta}\mathbb{E}_{\bm{o}\sim \tau, \bm{a}\sim\bm{\pi}, z\sim P_Z}[Q_i(\bm{o}, \bm{a}|z) - \alpha L_I(\bm{\pi}, U)], \label{eq: pg obj}
\end{eqnarray}
where 
$Q_i(\bm{o}, \bm{a}|z)$ is the centralized critic function,
and $\alpha>0$ is the hyperparameter for information maximization term.
$Q_i(\bm{o}, \bm{a}|z)$ can also be substituted with the advantage function used by COMA.
We do not share parameters among agents.
Note that $o_i$ may be a partial observation of agent $i$,
and additional communication mechanism can be introduced to ensure theoretical correctness of Eq. \eqref{eq: mi loss}.
However, we empirically show that in some partially observable environments,
e.g., particle worlds \cite{lowe2017multi}, 
where the agent can infer global state from its local observation,
SIC can still work with $o_i$.
When applied to Multi-agent Actor-Critic frameworks \cite{lowe2017multi}, 
the objective of updating critic remains unchanged.

\begin{figure}[t!]
	\centering
	\includegraphics[width=0.45\textwidth]{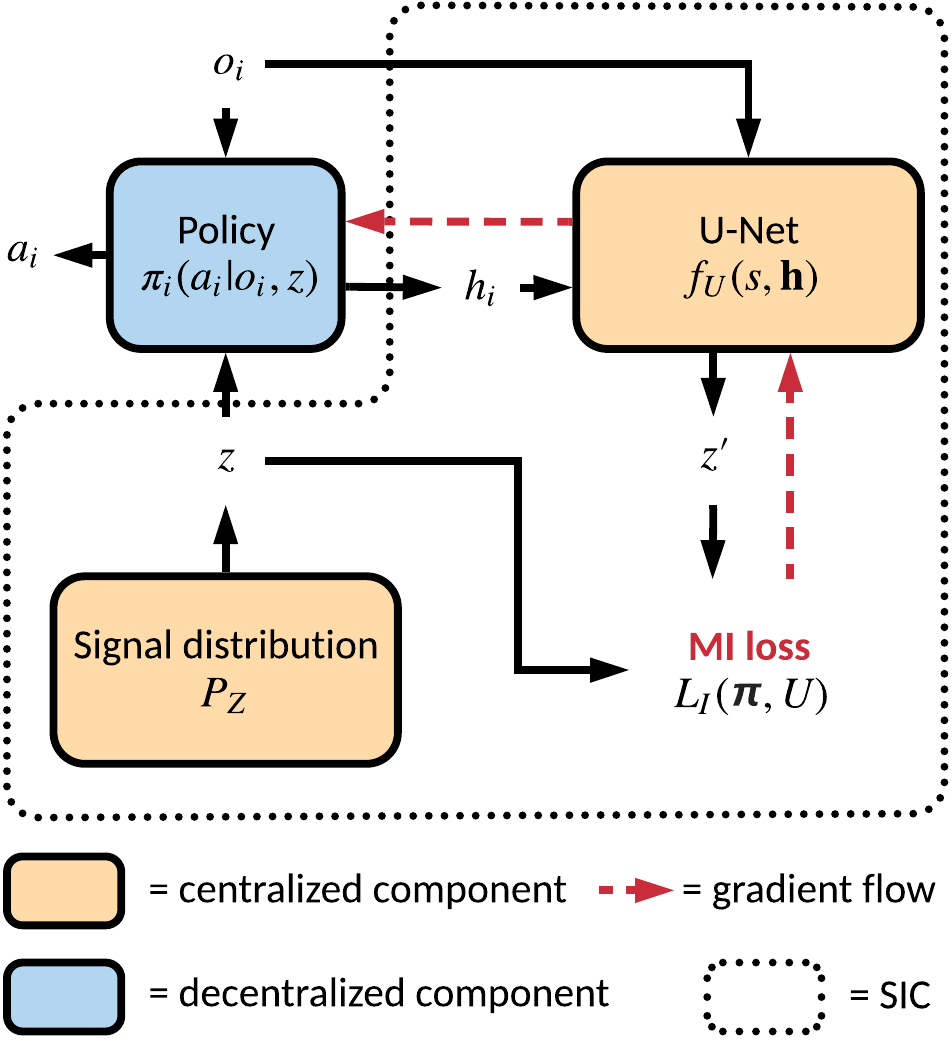}
	\caption{Illustration of SIC.
	Black arrows indicate how variables are passed between components.
	Note that in U-Net, 
	we use the concatenation of last hidden vectors in policy networks $h$, instead of $a$,
	to enable gradient flow when adopting stochastic policies.}
	\label{fig:sic}
\end{figure}

\section{Experiments}
In this section, we firstly evaluate SIC on a simple \textit{Rock-Paper-Scissors-Well} game to see whether SIC encourages coordination and thus leads to better performance. Then, we comprehensively study SIC in more challenging \textit{Predator-Prey} game under different scenarios.

\subsection{Rock-Paper-Scissors-Well}
We use a 2 vs. 2 variant of the matrix game, \textit{Rock-Paper-Scissors-Well} \cite{rpsw}, 
which is an extension to traditional \textit{Rock-Paper-Scissors}. 
Each team consists of two independent agents, and the available actions of each agent are $Access (A)$ and $Yield (Y)$. The joint action space consists of $(Y, Y)$, $(Y, A)$, $(A, Y)$, and $(A, A)$,  which can be named as $Rock$, $Paper$, $Scissors$ and $Well$. The former three actions play as in the traditional \textit{Rock-Paper-Scissors} game,  while \textit{Well} wins only against \textit{Paper} and is defeated by \textit{Rock} and \textit{Scissors}. The payoff matrix is presented in Table \ref{tab:reward matrix}.

Assume both teams are controlled by centralized controllers,
the best $\pi^C$ is
$P(Y,Y)=P(Y,A)=P(A,Y)=\frac{1}{3}$ and $P(A,A)=0$, 
since it is always better to take $(Y, A)$ instead of $(A, A)$.
To achieve this, 
agents within the same team need to coordinate to avoid the disadvantaged joint action $(A, A)$, 
and choose others uniformly randomly.
From a probabilistic perspective, the coordination requires high correlation between teammates, 
otherwise either $(A, A)$ is inevitable to appear as long as $\pi_1(A)\times \pi_2(A) > 0$, 
or the joint action space degenerates to $\{(Y,Y), (Y,A)\}$ or $\{(Y, Y), (A, Y)\}$.

\begin{table}[t!]
	\centering
	\caption{Payoff matrix of the 2 vs. 2 Rock-Paper-Scissors-Well game. 
	Both row and column players consist of two agents,
	who coordinate with individual actions as Y or A to play a joint action, e.g., Paper (Y, A),
	and receive shared rewards.
	Failed coordination will result in a bad joint policy vulnerable to the opponent team.}
	\label{tab:reward matrix}
	\begin{tabular}{cc|cccc}
		\toprule
		&& Rock & Paper & Scissors & Well \\
		&& Y, Y & Y, A & A, Y & A, A \\
		\midrule
		Rock & Y, Y & (0, 0) & (1, -1) & (-1, 1) & (1, -1) \\
		Paper & Y, A & (-1, 1) & (0, 0) & (1, -1) & (1, -1) \\
		Scissors & A, Y & (1, -1) & (-1, 1) & (0, 0) & (-1, 1) \\
		Well & A, A & (-1, 1) & (-1, 1) & (1, -1) & (0, 0) \\
		\bottomrule
	\end{tabular}
\end{table}

\begin{figure}[t!]
	\centering
	\begin{subfigure}{0.22\textwidth}
		\centering
		\includegraphics[width=\textwidth]{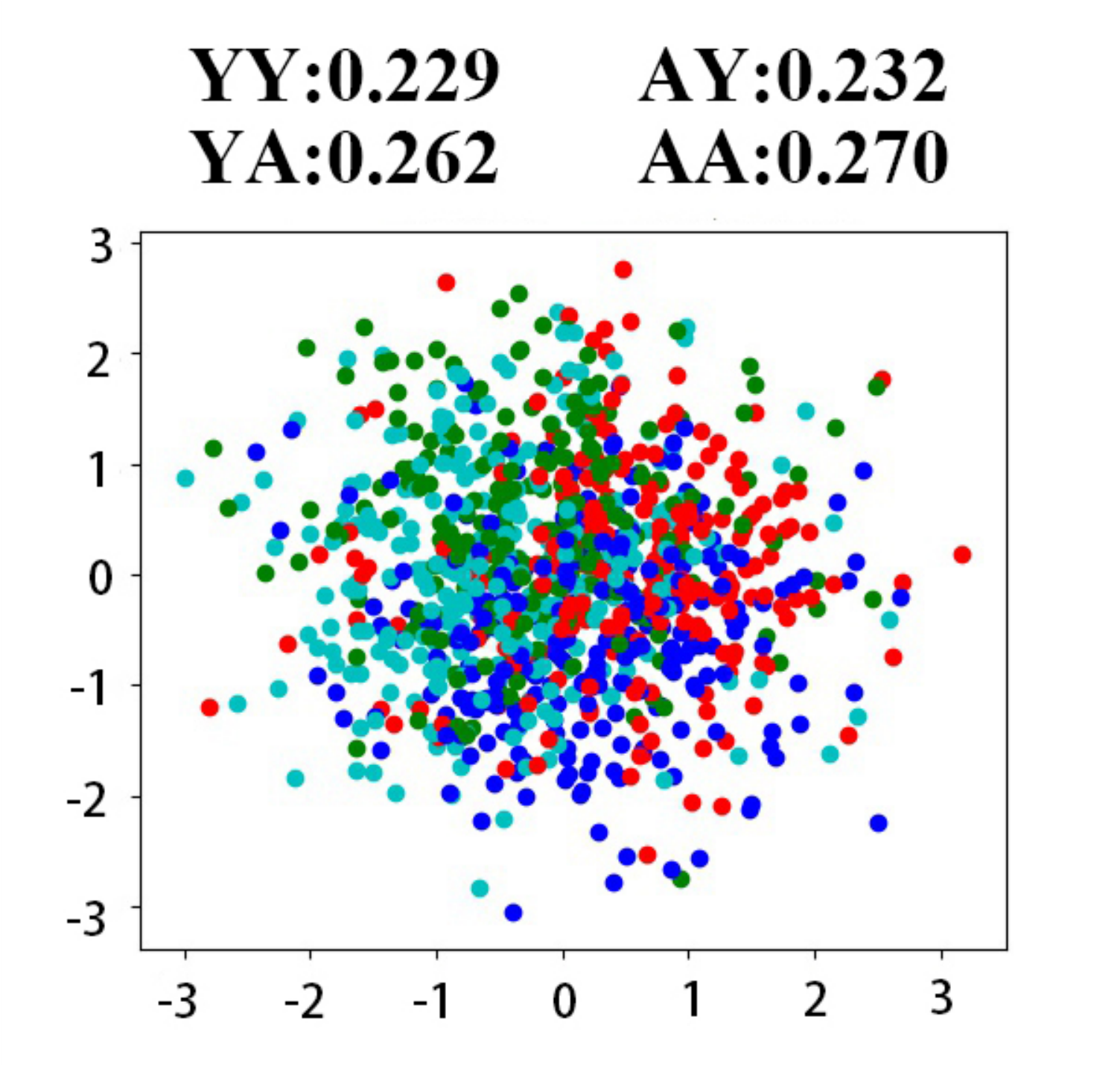}
		\subcaption{Row player (0 ep).}
	\end{subfigure}
	{\LARGE$\xrightarrow{}$}
	\begin{subfigure}{0.22\textwidth}
		\centering
		\includegraphics[width=\textwidth]{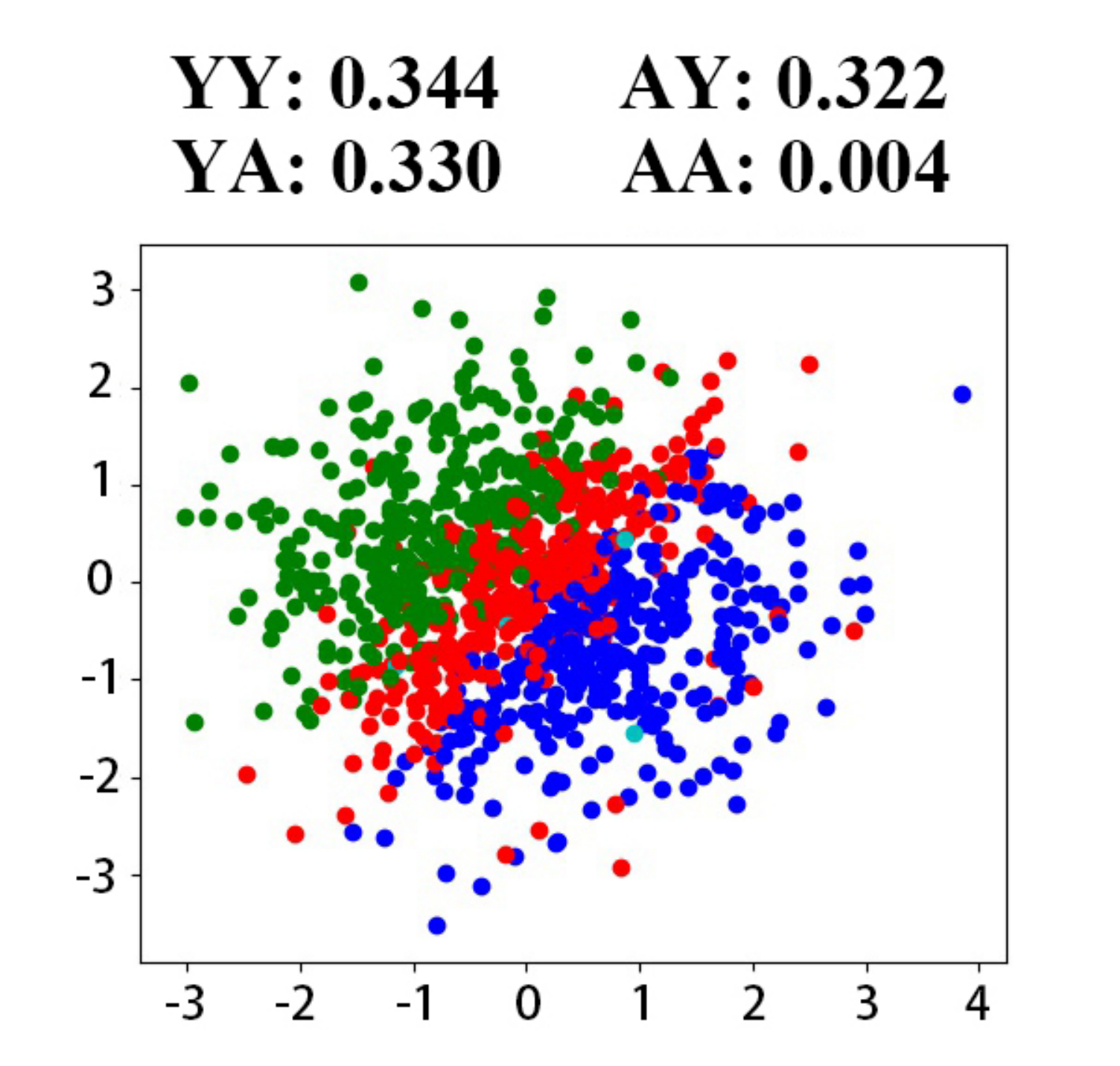}
		\subcaption{Row player (100k ep).}
	\end{subfigure}
	\\
	\begin{subfigure}{0.22\textwidth}
		\centering
		\includegraphics[width=\textwidth]{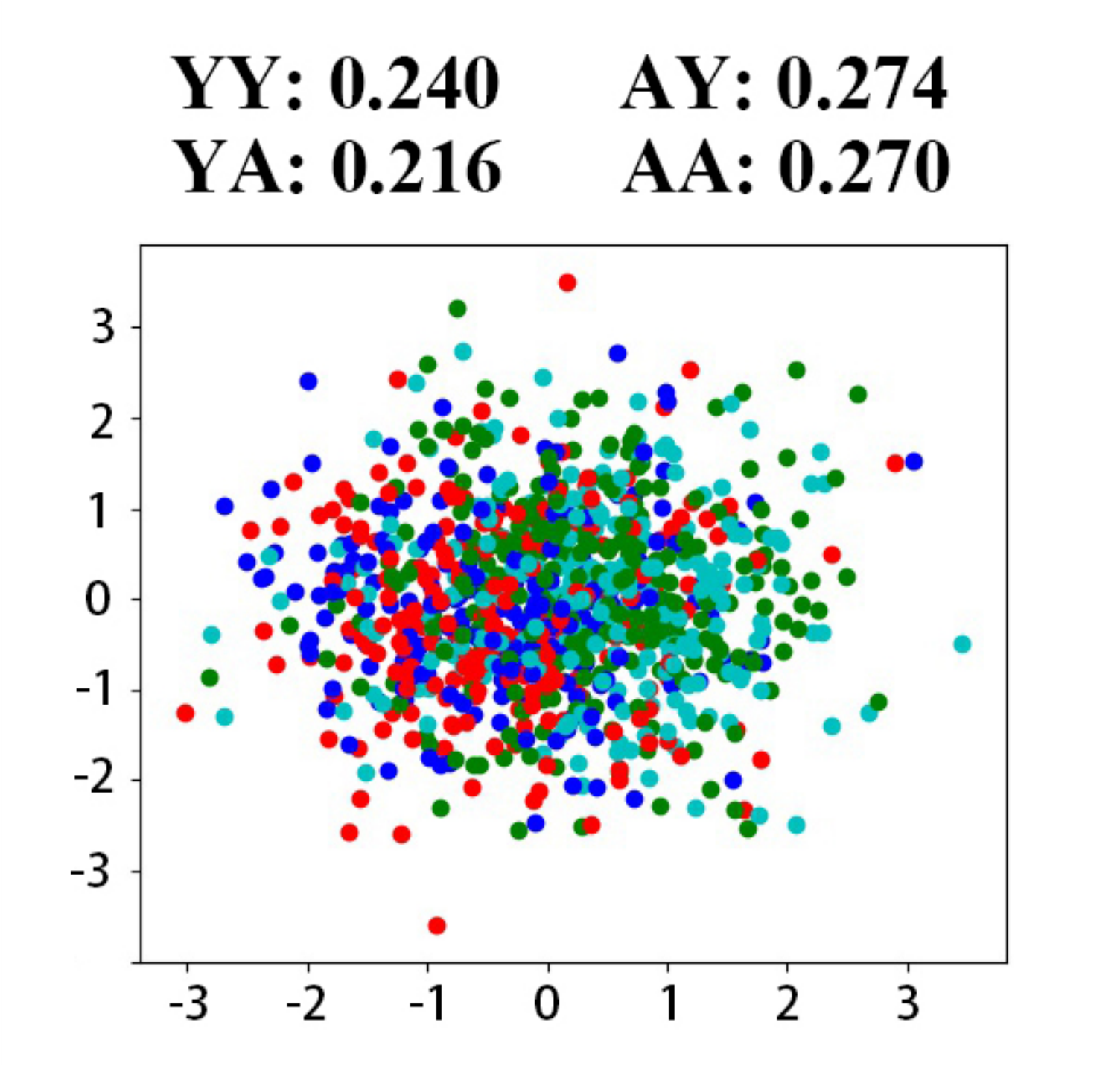}
		\subcaption{Col player (0 ep).}
	\end{subfigure}
	{\LARGE$\xrightarrow{}$}
	\begin{subfigure}{0.22\textwidth}
		\centering
		\includegraphics[width=\textwidth]{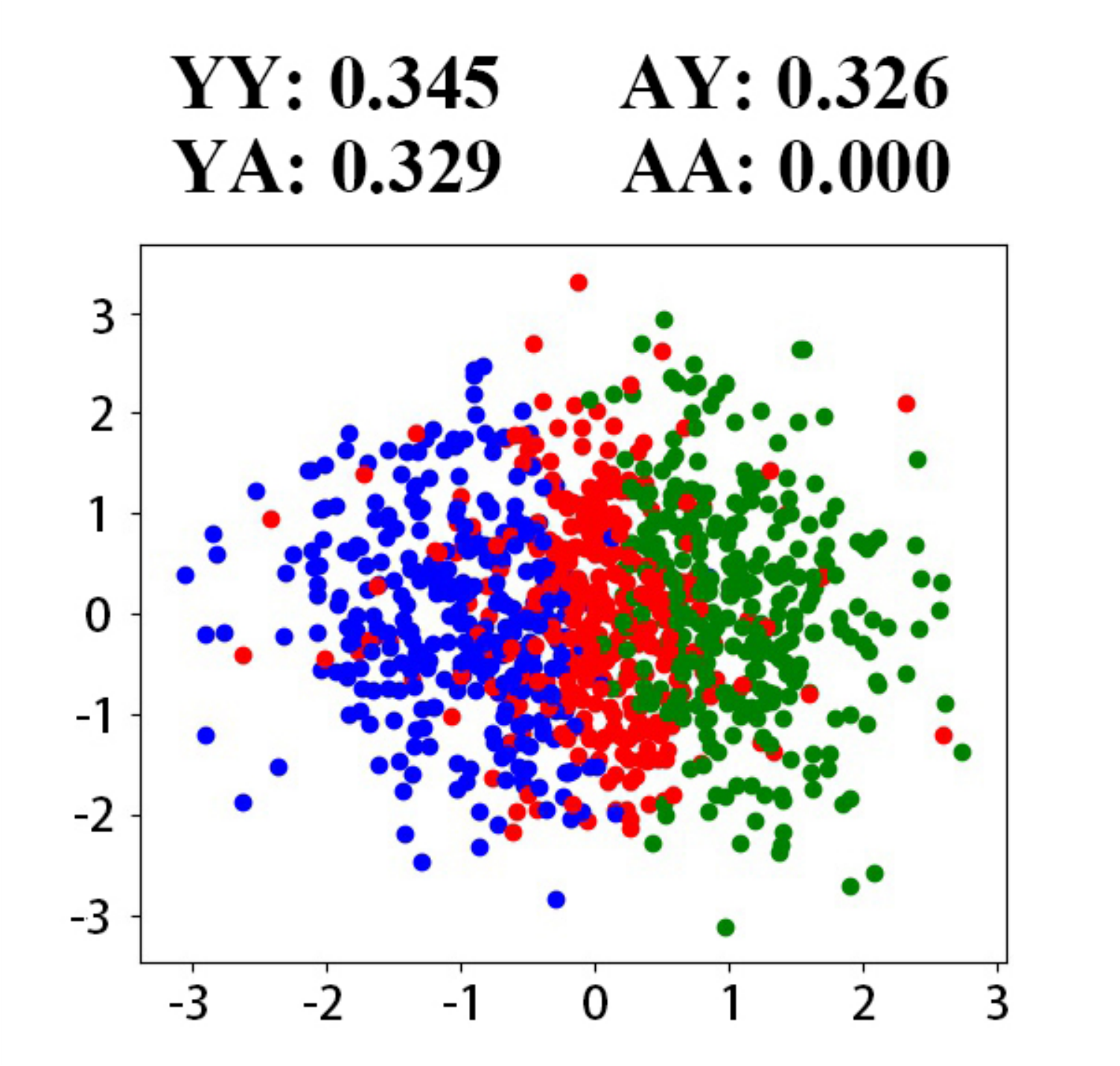}
		\subcaption{Col player (100k ep).}
	\end{subfigure}
	\caption{Correlations between signal distribution and joint actions. 
		Each point represents a 2-dim signal, and red, green, blue, and cyan represent the corresponding joint action as $(Y,Y)$, $(Y,A)$, $(A,Y)$, and $(A,A)$ respectively.
		The frequency of different actions in these 5000 points is shown above each sub-figure.}
	\label{fig:visualization_matrix_game}
\end{figure}

\subsubsection{One-step Matrix Game}
In the first scenario, we set the length of each episode as one.
We apply our SIC module to REINFORCE algorithm, and denote it as \textit{SIC-RE}.
The team-shared coordination signal $\ve{z}\in\mathbb{R}^2$ is sampled from $\mathcal{N}(\ve{0}, I)$.
Note that each agent takes a stochastic policy, 
and signals received by the two teams are different.
After training, the four agents converge to an equilibrium, 
with rewards of both row and column players equal to 0.
To better study which kind of equilibrium agents have reached, 
we randomly sample 5000 signals and test how agents respond to them before and after training.
Fig. \ref{fig:visualization_matrix_game} illustrates the joint policy of both teams. 
We can see that
\begin{enumerate}
	\item Before training (a\&c), 
	the frequency of each joint action is roughly 0.25 for both team, 
	since each agent takes a random individual policy.
	In addition, the distribution of signals triggering different joint actions are quite spreading.
	\item After training (b\&d), 
	the signal space is roughly divided into three ``zones'', 
	with each zone representing one joint action.
	The ``area'' of each zone, i.e., the probability of sampling one signal bolonging to the zone, is roughly $1/3$, 
	which indicates that the result is close to the best performance a centralized controller can achieve.
\end{enumerate}

\subsubsection{Multi-step Matrix Game}

In the second scenario, we set the length of the episode as 4 steps. 
We denote the matrix game in Table \ref{tab:reward matrix} as $M_4$, 
since we expect the fourth joint action to be deprecated by agents.
Then we generate a new matrix, $M_i$, 
by exchanging the fourth row with the $i$-th row, and the fourth column with the $i$-th column sequentially.
In this way we obtain a set of matrices $\mathcal{M}=\{M_1, M_2, M_3, M_4\}$.
We design a multi-step matrix game, 
where two teams play according to a random payoff matrix drawn from $\mathcal{M}$ in each step.
Each agent can only observe the ID $i\in\{1, 2, 3, 4\}$ of the current matrix and the coordination signal.
To simulate sparse rewards, 
we only give agents the sum of rewards on each step after an episode of game is finished,
and train them with discounted returns.

To evaluate SIC-RE, we use REINFORCE algorithm with fully independent agents as the baseline model, 
which we denote as \textit{IND-RE}. We conduct experiments and plot averaged reward curves of row players in Fig. \ref{fig:reward curves of multi-step matrix game}. 
We also plot joint policy curves of row players in $M_4$ (the same as Table \ref{tab:reward matrix}) of SIC-RE vs SIC-RE and IND-RE vs IND-RE in Fig. \ref{fig: example curves of multi-step matrix game} to study their reached equilibrium.
More curves can be found in Appendix \ref{sec: visualization}.
We can see that
\begin{enumerate}
	\item Even in multi-step scenarios with sparse rewards, 
	coordination signal can still coordinate agents as the centralized controller does.
	Fig. \ref{fig: example curves of multi-step matrix game}-a shows that joint policy of row players gradually converge to $P(Y,Y)=P(Y,A)=P(A,Y)=\frac{1}{3}$ and $P(A,A)=0$,
	and the response of row players is similar to Fig. \ref{fig:visualization_matrix_game}.
	\item
	Although IND-RE also reaches an equilibrium with a game value as 0 in IND-RE vs. IND-RE as shown in Fig. \ref{fig:reward curves of multi-step matrix game} and \ref{fig: example curves of multi-step matrix game}-b,
	its ability to coordinate is limited,
	as it can only formulate joint policy in $\Pi^D$.
	Therefore, in direct competition as SIC-RE vs IND-RE, 
	IND-RE is outperformed and stuck in a disadvantaged equilibrium with a negative averaged reward.
\end{enumerate}

\begin{figure}[t!]
	\centering
	\includegraphics[width=0.4\textwidth]{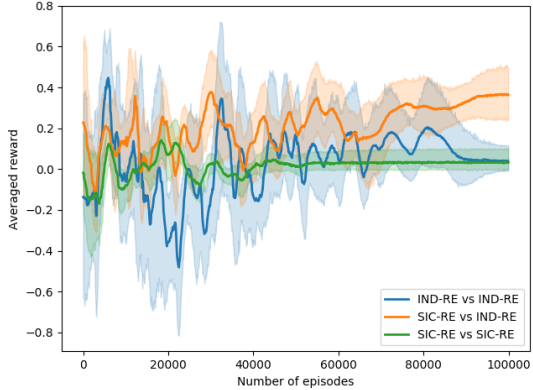}
	\caption{Average rewards of row players on 4-step matrix game. 
	    We report reward curves of (i) SIC-RE vs. SIC-RE, (ii) SIC-RE vs. IND-RE, and (iii) IND-RE vs. IND-RE with row players trained by the former model.
	}
	\label{fig:reward curves of multi-step matrix game}
\end{figure}

\begin{figure}[t!]
	\begin{subfigure}{0.42\textwidth}
		\centering
		\includegraphics[width=\textwidth]{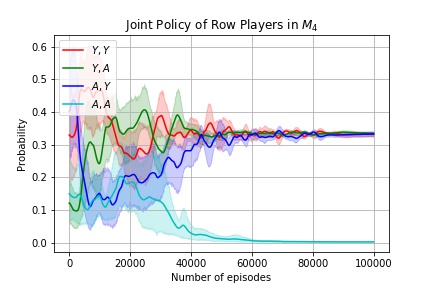}
		\subcaption{SIC-RE vs SIC-RE}
	\end{subfigure}
	
	\begin{subfigure}{0.42\textwidth}
		\centering
		\includegraphics[width=\textwidth]{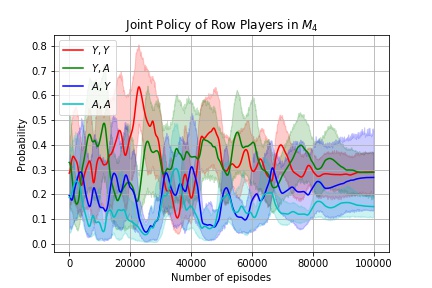}
		\subcaption{IND-RE vs IND-RE}
	\end{subfigure}
	\caption{Example joint policy curves of two settings in multi-step matrix game.}
	\label{fig: example curves of multi-step matrix game}
\end{figure}

\subsection{Predator-Prey}
In this section, we evaluate SIC on \textit{Predator-Prey} games \cite{benda1986optimal,matignon2012independent,lowe2017multi}, 
which is a common task to study coordination of multi-player games.  
We customize our Predator-Prey game based on the implementation in \cite{lowe2017multi}.
In this game, 
$M$ slow \textit{predators} and $M$ fast \textit{preys} are randomly placed in a two-dimensional world with $L=2$ large landmarks impeding the way,
and predators need to collaborate to collide with another team of agents, preys.
Each agent is depicted by several attributes including its coordinates and velocity while landmarks are depicted by the location only.
The observation of each agent includes its own attributes and egocentric attributes of other agents and landmarks. 
The action of each agent contains 4 moving directions and a stop action. 
Whenever a collision happens, 
all predators are equally rewarded while all preys are equally penalized immediately.

\begin{table*}[htbp]
	\centering
	\caption{Cross-comparison among different models.
	Results are reported in terms of predator scores, 
	which is proportional to the number of collisions.
	Each row shows results against the same prey;
	the higher the score, 
	the better the predator model,
	and the bold value is the highest one.
	Each column shows results against the same predator;
	the lower the score, 
	the better the prey model, 
	and the underlined value is the lowest one.
	}
	\label{tab:cross comparison}
	\begin{subtable}{\textwidth}
	\centering
	\subcaption{2 vs. 2 Predator-Prey game.}
	\begin{tabular}{c|ccccc}
		\toprule
		\backslashbox{Preys}{Predators} & COMA & MADDPG &SIC-COMA & SIC-MA &  SIC-MA (w/o $L_I$) \\
		\midrule
		COMA 	 &7.38 $\pm$ 2.58 & 132.27 $\pm$ 9.93 & 6.35 $\pm$ 1.87&\textbf{139.27} $\pm$ 7.45 & 133.63 $\pm$ 7.22 \\
		MADDPG	 &0.37 $\pm$ 0.14& 3.07 $\pm$ 0.65&0.56 $\pm$ 0.15&\textbf{3.32} $\pm$ 0.47 &3.14 $\pm$ 0.56 \\
		SIC-COMA  &6.76 $\pm$ 1.18&139.37 $\pm$ 13.38&5.44 $\pm$ 1.04&\textbf{145.11} $\pm$ 11.55&140.82 $\pm$ 15.11\\
		SIC-MA 	 &\underline{0.34} $\pm$ 0.14& \underline{2.81} $\pm$ 0.44&\underline{0.37} $\pm$ 0.14& \underline{\textbf{3.15}} $\pm$ 0.25&\underline{3.13} $\pm$ 0.44\\
		\bottomrule
	\end{tabular}
	\end{subtable}
	\vspace{0.2cm}
	
	\begin{subtable}{\textwidth}
	\centering
	\subcaption{4 vs. 4 Predator-Prey game.}
	\begin{tabular}{c|ccccc}
		\toprule
		\backslashbox{Preys}{Predators}  & COMA & MADDPG &SIC-COMA & SIC-MA &  SIC-MA (w/o $L_I$) \\
		\midrule
		COMA   &21.8 $\pm$ 3.3& 76.3 $\pm$ 12.7 &  25.3 $\pm$ 4.5 &\textbf{78.6} $\pm$ 13.1  & 75.9 $\pm$ 10.2 \\
		MADDPG	 & 21.1 $\pm$ 2.2 &41.3 $\pm$ 3.9&  21.9 $\pm$ 2.6 &\textbf{42.2} $\pm$ 4.7 &  39.6 $\pm$ 6.6\\
		SIC-COMA  & \underline{20.1} $\pm$ 2.0 & 57.3 $\pm$ 8.8  & 21.6 $\pm$ 2.7 & \textbf{58.2} $\pm$ 8.4 &  56.5 $\pm$ 9.2 \\
		SIC-MA 	 & 20.5 $\pm$ 1.9 & \underline{37.3} $\pm$ 3.7& \underline{21.2} $\pm$ 2.1 &  \underline{\textbf{41.5}} $\pm$ 5.2 & \underline{38.8} $\pm$ 6.0\\
		\bottomrule
	\end{tabular}
	\end{subtable}
\end{table*}

We use following baseline models as baselines:
\begin{enumerate}
    \item \textbf{COMA} \cite{foerster2018counterfactual}: 
    Each agent learns an individual policy, $\pi_i(a_i|\tau_i)$,
    where $\tau_i$ is the local action-observation history, 
    and a centralized action-value critic, $Q_i(\bm{o}, \bm{a})$. 
    The advantage function used when updating policy networks, $A_i(\bm{s}, a_i, \bm{a}_{-i})$, 
    is computed as $Q_i(s, \bm{a})-\sum_{a_i}\pi_i(a_i|\tau_i)Q_i(\bm{o}, \bm{a})$.
    \item \textbf{MADDPG} \cite{lowe2017multi}: 
    Each agent learns a deterministic individual policy, $\mu_i(a_i|o_i)$ and a centralized action-value critic, 
    $Q_i(\bm{o}, \bm{a})$.
    To make policies differentiable in the discrete action space, 
    we adopt the Gumbel-softmax trick as \cite{lowe2017multi} does.
\end{enumerate}
We integrate SIC with baselines and denote them as SIC-MA and SIC-COMA.
We use hyper-parameters of neural networks, 
like the number of units and activation functions,
in the original paper,
and inherit them in SIC variants.
To evaluate the performance of different models,
we conduct cross-comparison among them and report predator scores with 10 random seeds.
The results are shown in Table \ref{tab:cross comparison}, 
where the bold value is the highest score in a row, 
and the underlined value is the lowest score in a column.
We can see that: 
\begin{enumerate}
	\item SIC-MA significantly outperforms all other models, 
	both as predators and preys.
	In addition, the application of SIC presents stable improvements when integrated into different models.
	\item Ablation analysis shows that $L_I$ is a strong constraint to enforce agents to coordinate.
	Absence of $L_I$ is harmful to the performance of models, 
	although minor improvements may be observed.
\end{enumerate}

\begin{figure*}[t!]
	\centering
	\begin{subfigure}{0.31\textwidth}
		\centering
		\includegraphics[width=\textwidth]{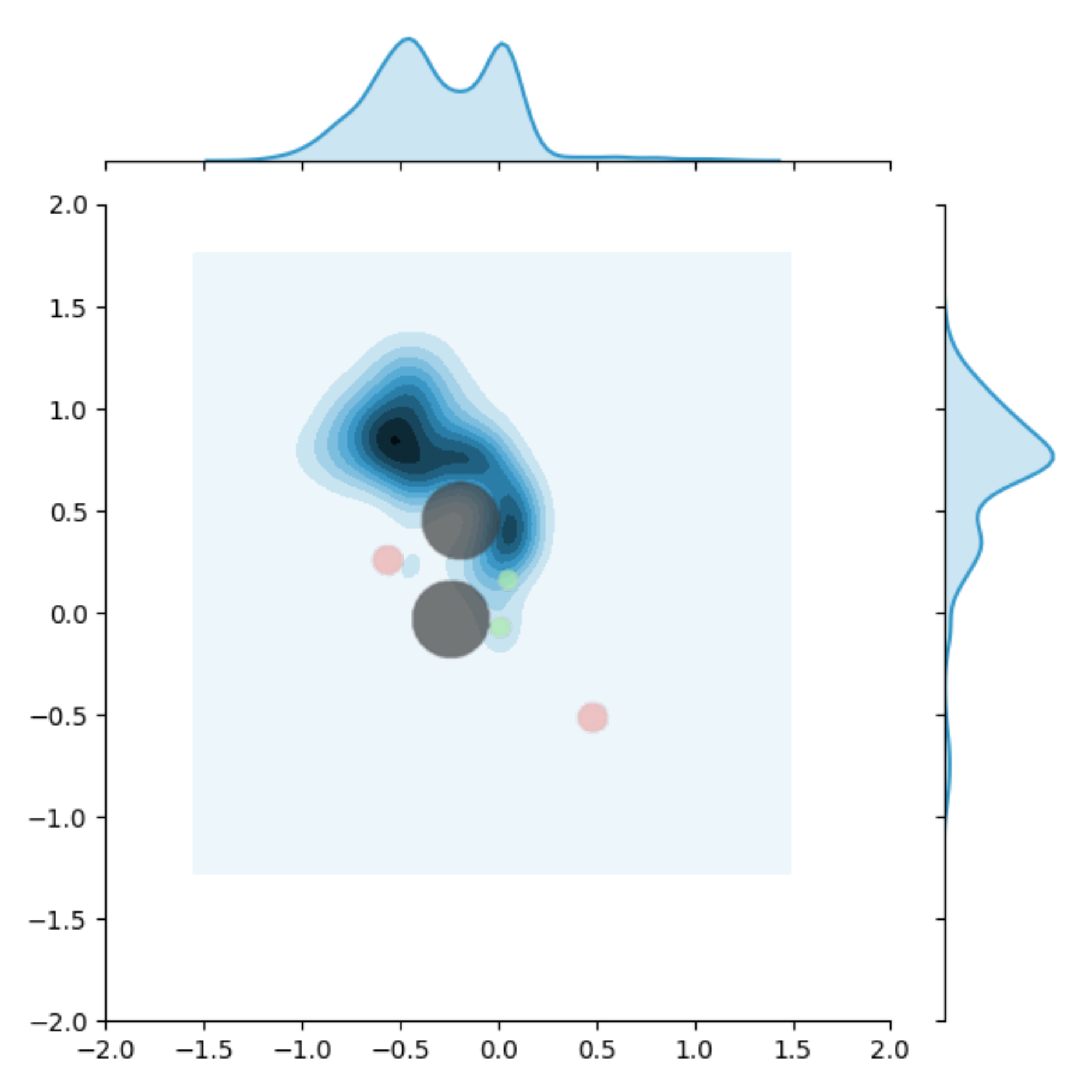}
		\subcaption{MADDPG vs MADDPG\\ 
		(3068 collisions)}
	\end{subfigure}
	\begin{subfigure}{0.31\textwidth}
		\centering
		\includegraphics[width=\textwidth]{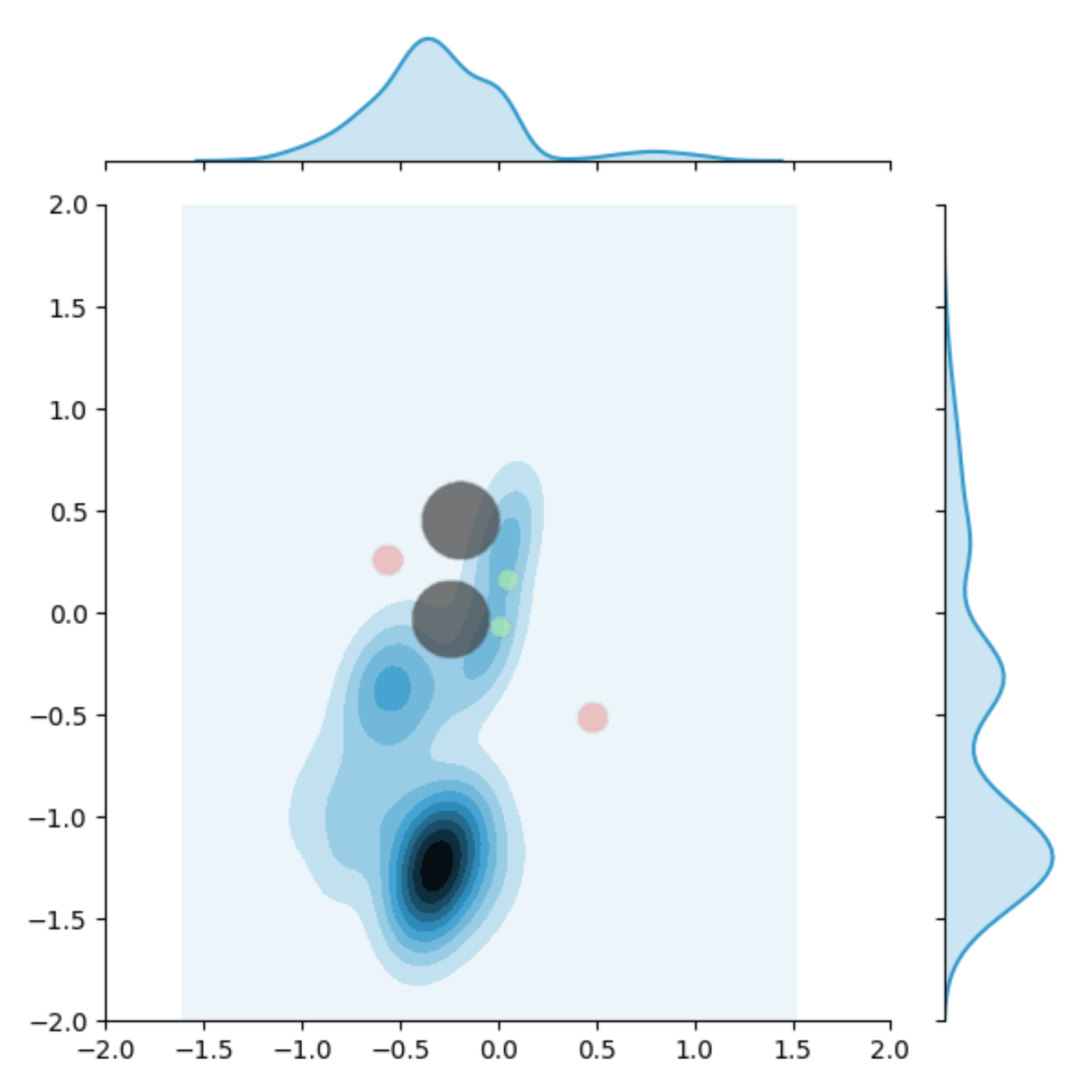}
		\subcaption{MADDPG vs SIC-MA\\ 
		(2041 collisions)}
	\end{subfigure}
	\begin{subfigure}{0.31\textwidth}
		\centering
		\includegraphics[width=\textwidth]{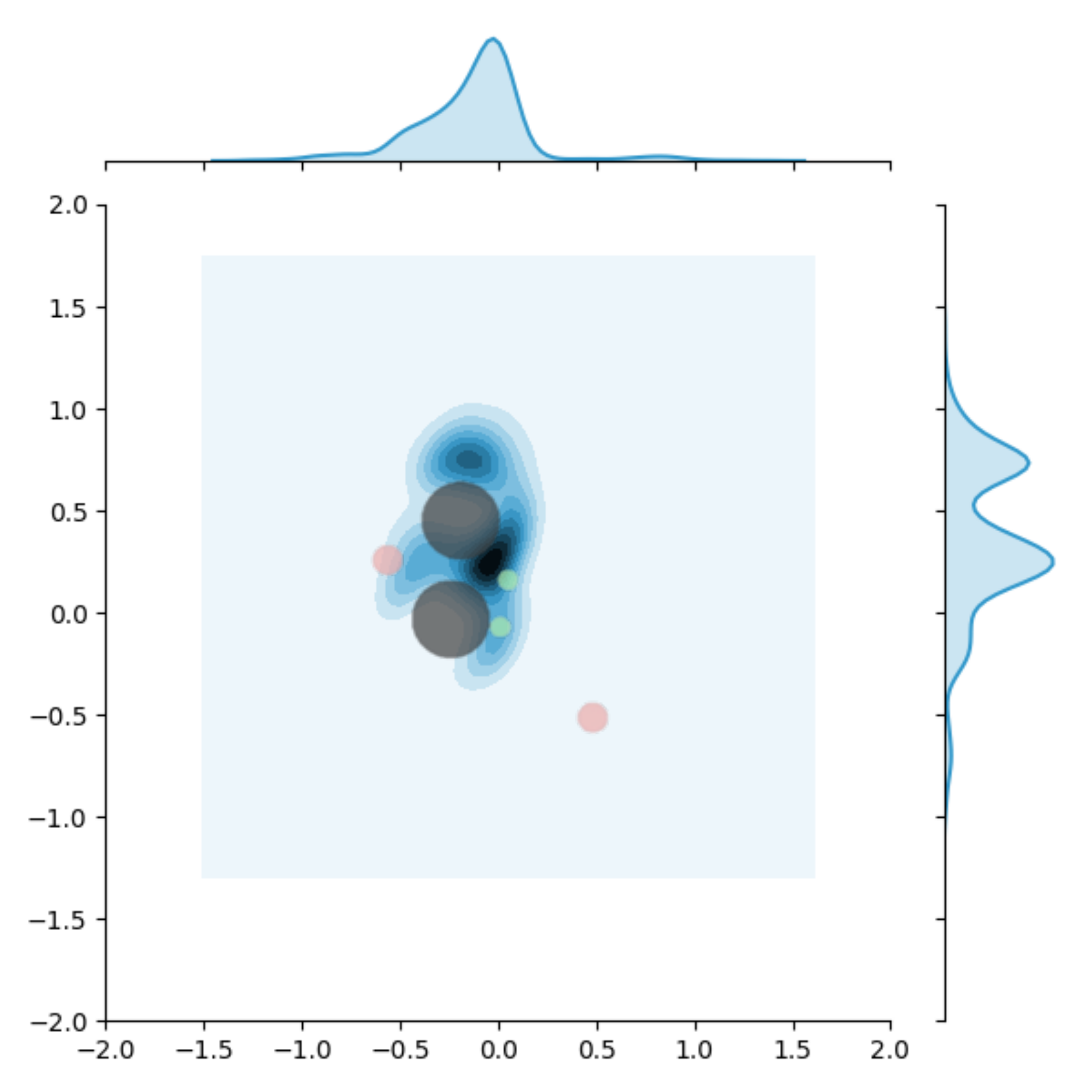}
		\subcaption{SIC-MA vs MADDPG\\ 
		(3806 collisions)}
	\end{subfigure}
	\caption{The density and marginal distribution of collision's positions, $(x, y)$, in 10000 repeated games.
		We repeat experiments from the same randomly generated environment, 
		with yellow, green and black circles representing predators, preys and landmarks at the start of each game.
		The result shows that compared to MADDPG,
		SIC-MA presents more diverse strategy and better performance both as predators and preys.
		}
	\label{fig:case study}
\end{figure*}

Besides comparing directly on the testing reward, we visualize the distribution of collision positions in Fig. \ref{fig:case study}.
We repeat MADDPG vs. MADDPG, SIC-MA vs. MADDPG and MADDPG vs. SIC-MA for 10000 games with different seeds. 
We reset each game to the same state in each episode, and collect positions of collisions in the total 250000 steps. We plot the density distribution of positions, $(x, y)$, marginal distributions of $x$ and $y$, and the initial particle world in Fig. \ref{fig:case study}. 
Data points $(x, y)$ with higher frequency in collected data are colored darker in the plane, 
and values of its components $x$ and $y$ in marginal distributions is higher.
Note that color and height reflect density instead of the number of times,
and cannot be compared between sub-figures directly.
In Fig. \ref{fig:case study}-a,  most collisions happen in the top or right parts around the upper landmark, which reflects that both predators and preys as running to the top-left corner. 
In Fig. \ref{fig:case study}-b, when SIC-MA controls preys, most collisions happen in the middle to the bottom of the whole map, 
showing that preys' trajectories cover a large area. 
When SIC-MA plays predators as in Fig. \ref{fig:case study}-c, 
collisions also appears in more diverse positions, 
and we observe that the right predator sometimes drive preys into the aisle between two landmarks, where the left predator is ambushing to capture them.
Specifically, 3806 and 2041 collisions happen when SIC-MA controls predators and preys respectively while 3068 collisions happen in MADDPG vs. MADDPG. In other words,  as predators, SIC captures more preys, and as preys, it avoids being captured more effectively than MADDPG. 
This evidences that SIC-MA learns better policies compared to MADDPG.

\begin{figure}[t!]
	\centering
	\includegraphics[width=0.35\textwidth]{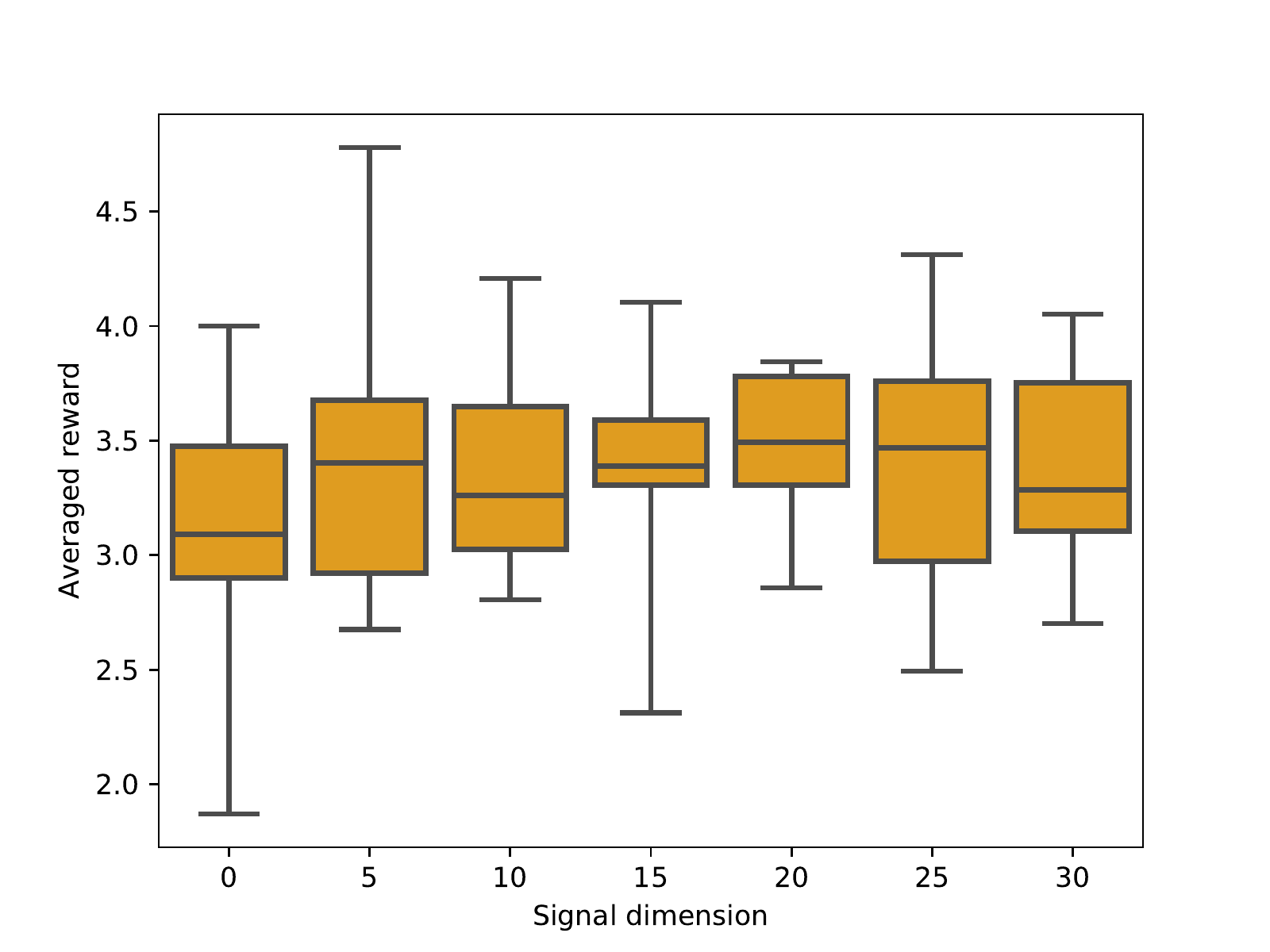}
	\caption{Results of parameter sensitivity analysis on $D_z$.}
	\label{fig:sensitivity}
\end{figure}

\subsection{Parameter Sensitivity} \label{sec: parameter sensitivity}
We conduct a parameter sensitivity analysis by testing SIC-MA vs MADDPG with different $D_z$ in 2 vs. 2 predator-prey environment, 
and report results in Fig. \ref{fig:sensitivity}. 
Note that SIC-MA with $D_z=0$ is equal to MADDPG,  and $D_z=20$ is the setting used in Table \ref{tab:cross comparison}. We can see that SIC-MA presents a stable improvement over MADDPG.
Most importantly, 
empirical results show that approximation through neural networks can compress $\mathcal{Z}$ and ensure good performance. 

\section{Related Works}

Unlike previous studies \cite{busoniu2008comprehensive} on MARL that adopt tabular methods and focus on coordination in simple environments,
recent works adopt deep reinforcement learning framework \cite{li2017deep} and turn to complex scenarios with high dimensional state and action spaces like particle worlds \cite{lowe2017multi} and StarCraft II \cite{vinyals2017starcraft}.
Among different approaches to model the controlling of agents, \textit{centralized training with decentralized execution} \cite{oliehoek2008optimal,lowe2017multi} outperforms others for circumventing the exponential growth of joint action space and the non-stationary environment problem \cite{li2017deep}.
Emergent communication \cite{lowe2019pitfalls} is proposed to enhance coordination and training stability.
Communication allows agents to pass messages between agents and ``share'' their observations via communication vectors.
It is essential to cooperative referential games \cite{foerster2016learning,lowe2017multi}, or the non-situated environment \cite{wagner2003progress}, 
where each agent has its own specialization as \textit{speaker} or \textit{listener} and the exchange of information is crucial to the  games.
A more widely-studied type of environments is the situated environment, 
where agents have similar roles and non-communicative actions.
\cite{sukhbaatar2016learning,peng2017multiagent} design special architectures to share information among all agents.
The {noisy channel} problem arises when all other agents use the same communication channel to send information simultaneously, 
and the agent needs to distinguish useful information from useless or irrelevant noise.
To alleviate this problem, 
\cite{jiang2018learning,das2018tarmac,iqbal2018actor} propose to introduce the attention mechanism to control the bandwidth of different agents dynamically.
However, 
communication requires large bandwidth to exchange information,
and the effectiveness of communication is under question as discussed by \cite{lowe2019pitfalls}.

The coordination problem \cite{boutilier1999sequential}, or the Pareto-Selection problem \cite{matignon2012independent}, has been discussed by a series of works in fully cooperative environments.
The solution to the coordination problem requires strong coordination among agents, i.e., 
all agents act as if in a fully centralized way.
In the game theory domain, it can also be viewed as pursuing \textit{Correlated equilibrium} (CE) \cite{leyton2008essentials,aumann1974subjectivity}, 
where agents make decisions following instructions from a correlation device.
It is desired that agents in the system can establish correlation protocols through adaptive learning method instead of constructing a correlation device manually for specific tasks
\cite{greenwald2003correlated} proposes to replace the value function in Q-learning with a new one reflecting agents' rewards according to some CE.
\cite{zhang2013coordinating} maintains coordination sets and select coordinated actions within these sets.
Apart from these methods, centralized signal is adopted by a variety of works \cite{cigler2011reaching,cigler2013decentralized,farina2018ex}.

Mutual information measures the mutual dependence between two variables, 
and has been used to enforce an information theoretic regularization by a variety of works in different domains \cite{barber2003algorithm,li2017infogail,chen2016infogan,eysenbach2018diversity}.
\cite{li2017infogail,chen2016infogan} use it to model the relationship between latent codes and outputs in generative models.
\cite{eysenbach2018diversity} proposes to substitute reward function with a mutual information objective to train policies unsupervisedly in the single-agent RL domain.
It maximizes mutual information between signal and state, 
and focuses on the diversity of the learned policy.

A similar concept to our coordination signal is \textit{common knowledge}, which refers to common information, 
e.g., representations of states, among partially observable agents.
Common knowledge is used to enhance coordination \cite{thomas2014psychology,foerster2018multi} and combined with communication \cite{korkmaz2014collective}.
Among them, \cite{foerster2018multi} proposes MACKRL which introduces a random seed as part of common knowledge to guide a hierarchical policy tree.
To avoid exponential growth of model complexity, MACKRL restricts correlation to pre-defined patterns, e.g., a pairwise one, 
which is too rigid for complex tasks.

\section{Conclusions}
We present the drawback of popular decentralized execution framework, 
and propose a signal instructed paradigm, 
which theoretically can coordinate decentralized agents as manipulated by a centralized controller.
We propose Signal Instructed Coordination (SIC), 
a novel module to enhance coordination of agents' policies in centralized training with decentralized execution framework,
SIC instructs agents by sampling and sending common signals to cooperative agents,
and incentivize their coordination by enforcing a mutual information regularization.
Our analysis show with the help of SIC, 
the joint policy of decentralized agents demonstrates better performance.


\newpage
\bibliographystyle{ACM-Reference-Format}  
\bibliography{sample-bibliography}  

\newpage
\appendix
\section{Experimental Results}

\subsection{2v2 Predator-Prey Experiment}

In 2v2 Predator-Prey experiment, we adopt 4 different models: MADDPG, SIC-MADDPG, COMA, SIC-COMA.
For both MADDPG and SIC-MADDPG, we use the same hyper-parameters with the original MADDPG paper except for the learning rate which is set to be $0.001$ for MADDPG and $0.0005$ for SIC-MADDPG. In SIC-MADDPG, we add a 20-dimensional signal, a U-Net with a ReLU MLP with 64 hidden units, and set the coefficient of MI loss to be $0.0001$. 
In COMA, we use the Adam optimizer with a learning rate of $0.00005$.
Gradient clipping is set to be 0.1. Both Actor and Critic are parameterized by a two-layer ReLU MLP with 64 units per layer which is the same with MADDPG.
We adopt GAE with $\gamma=0.99$ and $\lambda=0.8$. 
We use a batch size of $1000$. 
SIC-COMA adopts the same hyper-parameters as in COMA and the same signal, U-Net and MI loss with SIC-MADDPG.
For all models, we train with 10 random seeds.

\subsection{4v4 Predator-Prey Experiment}

In 4v4 Predator-Prey experiment, We adopt the same models and parameters with 2v2 case, with the modifications as follows:
The learning rates of MADDPG and SIC-MADDPG are both set to be $0.0005$. 
The coefficient of MI loss in SIC-MADDPG is set to be 0.01 in 4v4 case.

\subsection{Matrix Game Experiment}

We conduct three multi-step matrix game experiments with SIC-RE and IND-RE models.
For both models, we use the Adam optimizer with a learning rate of $0.0001$. 
The policy network is parameterized by a one-layer ReLU MLP with 8 hidden units.
For SIC-RE models, we use, a two-layer ReLU MLP with 8 hidden units as U-Net, and set the coefficient of MI loss to be $0.01$.

\section{Visualization for Joint Policy of Multi-step Matrix Game} \label{sec: visualization}
We plot the curves of joint policies of both row players and column players in multi-step matrix games in Fig. \ref{fig: sic_vs_sic}, \ref{fig: sic_vs_ind}, and \ref{fig: ind_vs_ind}.

\begin{figure*}[htb]
	\centering
	\begin{subfigure}{0.42\textwidth}
		\centering
		\includegraphics[width=\textwidth]{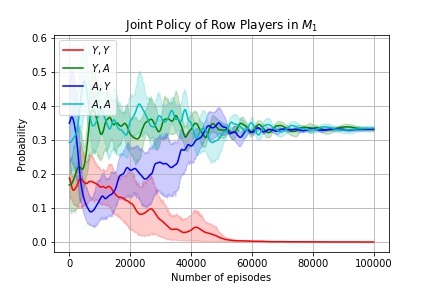}
	\end{subfigure}
	\begin{subfigure}{0.42\textwidth}
		\centering
		\includegraphics[width=\textwidth]{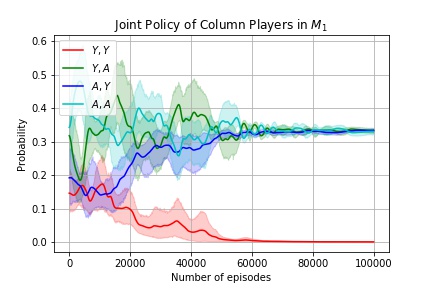}
	\end{subfigure}

	\begin{subfigure}{0.42\textwidth}
		\centering
		\includegraphics[width=\textwidth]{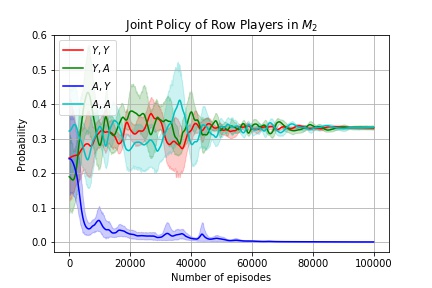}
	\end{subfigure}
	\begin{subfigure}{0.42\textwidth}
	\centering
	\includegraphics[width=\textwidth]{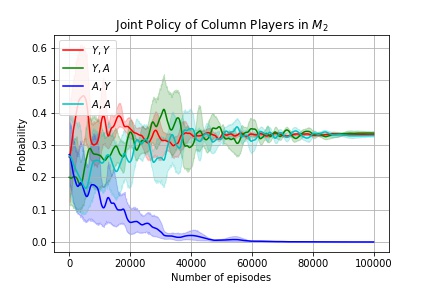}
	\end{subfigure}

	\begin{subfigure}{0.42\textwidth}
		\centering
		\includegraphics[width=\textwidth]{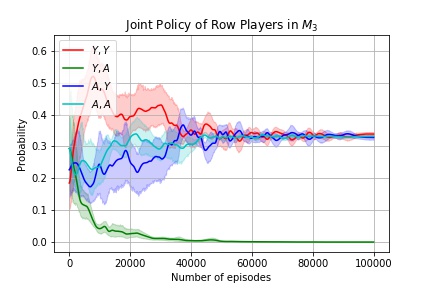}
	\end{subfigure}
	\begin{subfigure}{0.42\textwidth}
		\centering
		\includegraphics[width=\textwidth]{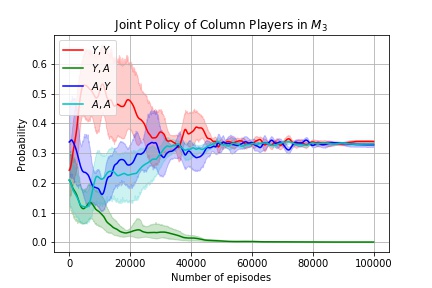}
	\end{subfigure}

\begin{subfigure}{0.42\textwidth}
	\centering
	\includegraphics[width=\textwidth]{SIC-RE_vs_SIC-RE_Row_Players_M4.jpg}
\end{subfigure}
	\begin{subfigure}{0.42\textwidth}
		\centering
		\includegraphics[width=\textwidth]{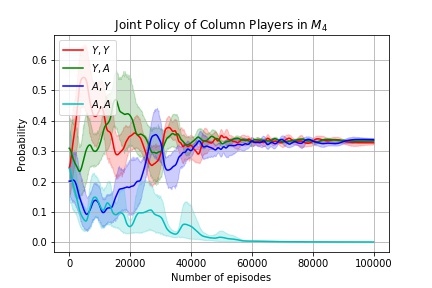}
	\end{subfigure}
	\caption{Joint Policy of SIC-RE vs SIC-RE. 
		During training, 
		the $i$-th joint action in $M_i$ is deprecated gradually,
		and all other joint actions are sampled uniformly randomly.}
	\label{fig: sic_vs_sic}
\end{figure*}

\begin{figure*}[htb]
	\centering
	\begin{subfigure}{0.42\textwidth}
		\centering
		\includegraphics[width=\textwidth]{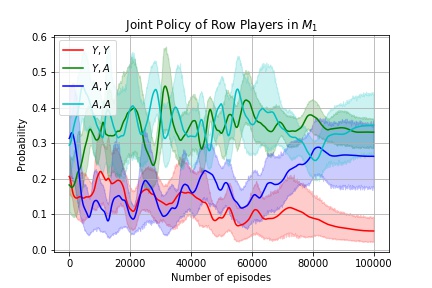}
	\end{subfigure}
	\begin{subfigure}{0.42\textwidth}
	\centering
	\includegraphics[width=\textwidth]{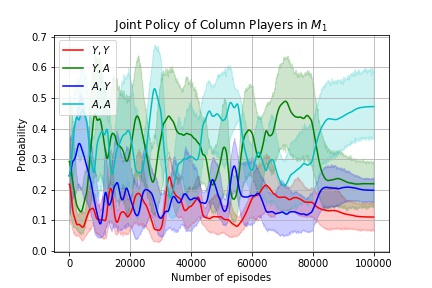}
\end{subfigure}

	\begin{subfigure}{0.42\textwidth}
		\centering
		\includegraphics[width=\textwidth]{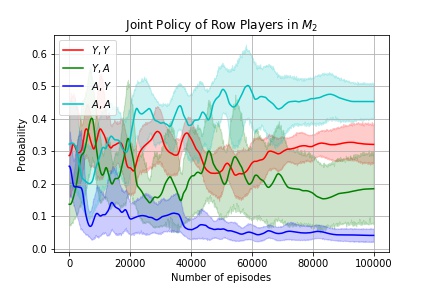}
	\end{subfigure}
\begin{subfigure}{0.42\textwidth}
	\centering
	\includegraphics[width=\textwidth]{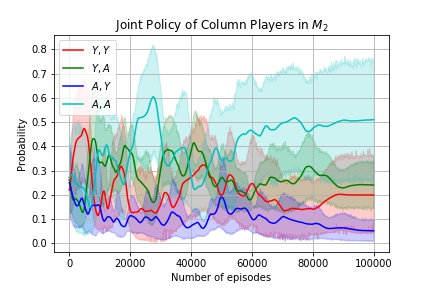}
\end{subfigure}

	\begin{subfigure}{0.42\textwidth}
		\centering
		\includegraphics[width=\textwidth]{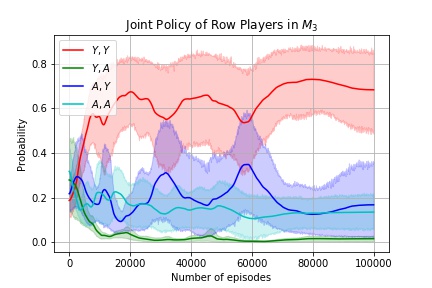}
	\end{subfigure}
\begin{subfigure}{0.42\textwidth}
	\centering
	\includegraphics[width=\textwidth]{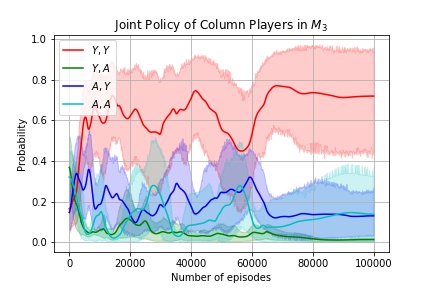}
\end{subfigure}

	\begin{subfigure}{0.42\textwidth}
		\centering
		\includegraphics[width=\textwidth]{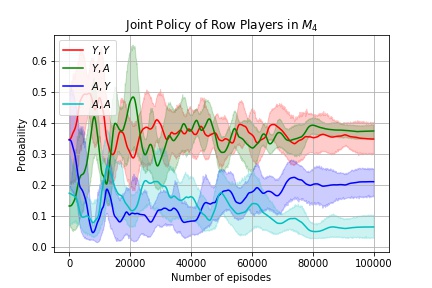}
	\end{subfigure}
	\begin{subfigure}{0.42\textwidth}
		\centering
		\includegraphics[width=\textwidth]{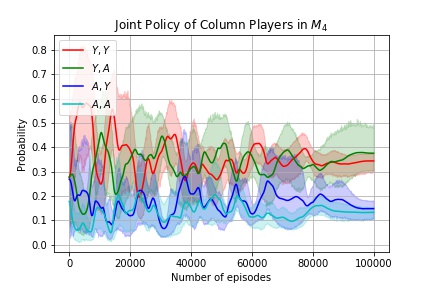}
	\end{subfigure}
	\caption{Joint Policy of SIC-RE vs IND-RE. 
		SIC-RE adjusts its joint policy to counter that of IND-RE, 
		and achieves a positive game value. }
	\label{fig: sic_vs_ind}
\end{figure*}

\begin{figure*}[htb]
	\centering
	\begin{subfigure}{0.42\textwidth}
		\centering
		\includegraphics[width=\textwidth]{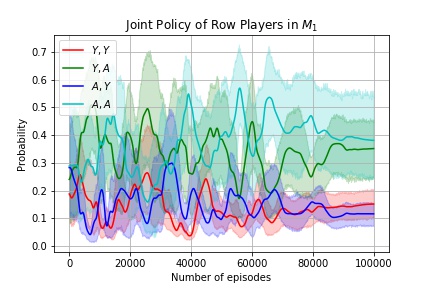}
	\end{subfigure}
	\begin{subfigure}{0.42\textwidth}
		\centering
		\includegraphics[width=\textwidth]{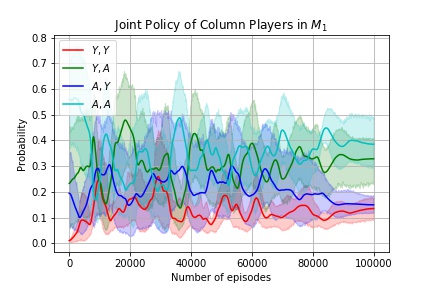}
	\end{subfigure}
	
	\begin{subfigure}{0.42\textwidth}
		\centering
		\includegraphics[width=\textwidth]{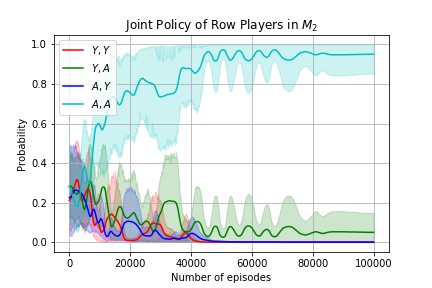}
	\end{subfigure}
	\begin{subfigure}{0.42\textwidth}
		\centering
		\includegraphics[width=\textwidth]{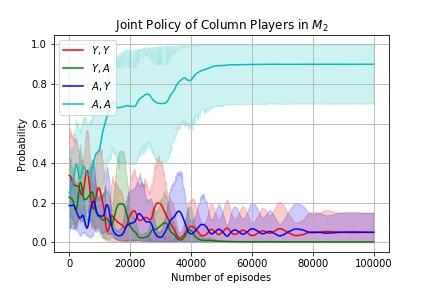}
	\end{subfigure}
	
	\begin{subfigure}{0.42\textwidth}
		\centering
		\includegraphics[width=\textwidth]{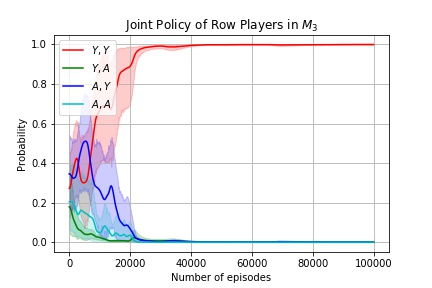}
	\end{subfigure}
	\begin{subfigure}{0.42\textwidth}
		\centering
		\includegraphics[width=\textwidth]{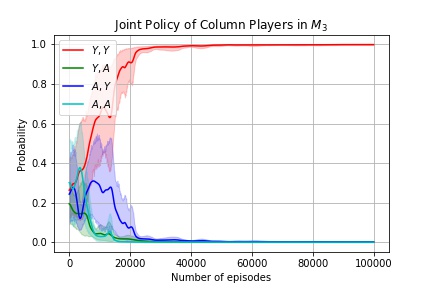}
	\end{subfigure}\
	
	\begin{subfigure}{0.42\textwidth}
		\centering
		\includegraphics[width=\textwidth]{IND-RE_vs_IND-RE_Row_Players_M4.jpg}
	\end{subfigure}
	\begin{subfigure}{0.42\textwidth}
		\centering
		\includegraphics[width=\textwidth]{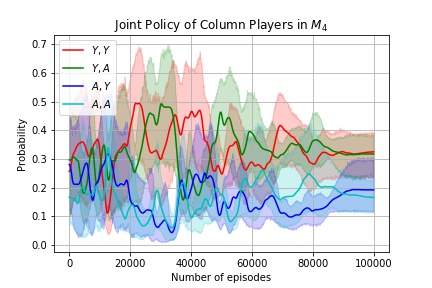}
	\end{subfigure}
	\caption{Joint Policy of IND-RE vs IND-RE.
		IND-RE only finds worse joint policy in the team-policy space, and 
		in some cases ($M_2$ and $M_3$), players play only one kind of joint action.}
	\label{fig: ind_vs_ind}
\end{figure*}

\end{document}